\documentclass[a4paper,10pt]{article}
\usepackage[bbgreekl]{mathbbol}
\usepackage{soul}
\usepackage[small]{titlesec}
\usepackage[utf8]{inputenc}
\usepackage[T1]{fontenc}  
\usepackage[]{mdframed}
\usepackage{mathrsfs}
\usepackage{caption}
\usepackage[noerroretextools,backend=biber,style=alphabetic,natbib=false,giveninits=true,doi=true,isbn=false,url=false,date=year,maxbibnames=99,sorting=nty,defernumbers=true]{biblatex}
\DeclareNameAlias{default}{family-given}

\DeclareFieldFormat[article]{title}{#1}
\renewbibmacro{in:}{}
\DeclareFieldFormat[article]{pages}{#1}
\DeclareFieldFormat{journal}{#1}
\renewcommand\bf\bfseries

\renewbibmacro*{volume+number+eid}{%
	\printfield{volume}%
	\setunit*{\addnbspace}
	\printfield{number}%
	\setunit{\addcomma\space}%
	\printfield{eid}}
\DeclareFieldFormat[article]{number}{\mkbibparens{#1}}
\DeclareFieldFormat[article]{volume}{\textbf{#1}}
\DeclareFieldFormat{year}{\mkbibparens{#1}}
\DeclareBibliographyDriver{article}{%
	\printnames{author}:%
	\newunit\newblock
	\printfield{title}%
	\newunit\newblock
	\printfield{journaltitle}
	\newunit
	\iffieldundef{number}{\printfield{volume}}{\printfield{volume}\addspace\printfield{number}}
	\addcomma\addspace\printfield{pages}\addspace
	\printfield{year}
}
\AtEveryBibitem{\clearfield{month}}
\AtEveryBibitem{\clearfield{day}}
\usepackage[english]{babel}
\usepackage[T1]{fontenc}
\usepackage{csquotes}
\usepackage{bbm}
\usepackage[leqno]{amsmath}
\usepackage{amsfonts,amsthm,amsbsy,amssymb,dsfont,stmaryrd}
\usepackage[dvipsnames]{xcolor}
\usepackage{braket}

\usepackage{caption}
\usepackage{subcaption}

\usepackage{enumitem}

\makeatletter
\newcommand{\leqnomode}{\tagsleft@true\let\veqno\@@leqno}
\newcommand{\reqnomode}{\tagsleft@false\let\veqno\@@eqno}
\makeatother

\usepackage{mathtools}
\usepackage[makeroom]{cancel}

\numberwithin{equation}{section}

\newcommand\myshade{85}
\colorlet{mylinkcolor}{violet}
\colorlet{mycitecolor}{YellowOrange}
\colorlet{myurlcolor}{Aquamarine}

\usepackage[left=2.5cm,right=2.5cm,top=2.5cm,bottom=2.5cm]{geometry}
\usepackage[unicode=true,pdfusetitle,
bookmarks=true,bookmarksnumbered=false,bookmarksopen=false,
breaklinks=false,pdfborder={0 0 1},backref=false,
linkcolor  = mylinkcolor!\myshade!black,
citecolor  = mycitecolor!\myshade!black,
urlcolor   = myurlcolor!\myshade!black,
colorlinks = true,
]
{hyperref}
\usepackage[nameinlink]{cleveref}

\usepackage{ifthen}

\usepackage{graphicx}
\usepackage{caption}
\usepackage{slashed}
\definecolor{ct_black}{HTML}{000000}
\definecolor{ct_orange}{HTML}{ED872D}
\definecolor{ct_purple}{HTML}{7A68A6}
\definecolor{ct_blue}{HTML}{348ABD}
\definecolor{ct_turquoise}{HTML}{188487}
\definecolor{ct_red}{HTML}{E32636}
\definecolor{ct_pink}{HTML}{CF4457}
\definecolor{ct_green}{HTML}{467821}

\definecolor{ct2_green}{HTML}{9FF781}
\definecolor{ct2_green_dark}{HTML}{088A08}

\theoremstyle{plain}
\newtheorem{thm}{\protect\theoremname}[section]
\theoremstyle{plain}
\newtheorem{lem}[thm]{\protect\lemmaname}
\theoremstyle{plain}
\newtheorem{cor}[thm]{\protect\corollaryname}
\theoremstyle{plain}
\newtheorem{prop}[thm]{\protect\propositionname}
\theoremstyle{plain}

\theoremstyle{remark}

\theoremstyle{definition}
\newtheorem{defn}[thm]{\protect\definitionname}

\providecommand{\assumptionname}{Assumption}
\providecommand{\claimname}{Claim}
\providecommand{\corollaryname}{Corollary}
\providecommand{\definitionname}{Definition}
\providecommand{\lemmaname}{Lemma}
\providecommand{\propositionname}{Proposition}
\providecommand{\remarkname}{Remark}
\providecommand{\theoremname}{Theorem}
\providecommand{\examplename}{Example}

\crefname{section}{Section}{Sections}
\crefname{example}{Example}{Examples}
\crefname{appendix}{Appendix}{Appendices}
\crefname{figure}{Figure}{Figures}
\crefname{assumption}{Assumption}{Assumptions}
\crefname{thm}{Theorem}{Theorems}
\crefname{lem}{Lemma}{Lemmas}
\crefformat{equation}{(#2#1#3)}
\crefname{table}{Table}{Tables}

\crefrangelabelformat{equation}{(#3#1#4--#5#2#6)}

\crefmultiformat{equation}{(#2#1#3}{, #2#1#3)}{#2#1#3}{#2#1#3}
\Crefmultiformat{equation}{(#2#1#3}{, #2#1#3)}{#2#1#3}{#2#1#3}

\newtheorem*{lem*}{\protect\lemmaname}

\newcommand{\ee}{\operatorname{e}}

\newcommand{\bS}{\mathbb{S}}

\newcommand{\NN}{\mathbb{N}}
\newcommand{\RR}{\mathbb{R}}

\newcommand{\PP}{\mathbb{P}}
\newcommand{\EE}{\mathbb{E}}

\newcommand{\calO}{\mathcal{O}}
\newcommand{\calR}{\mathcal{R}}

\newcommand{\Ord}[1]{\calO\left(#1\right)}

\newcommand\norm[1]{\left\lVert#1\right\rVert}
\newcommand{\ip}[2]{\langle #1, #2 \rangle}

\newcommand{\dif}{\operatorname{d}\!} 
\newcommand{\tr}{\operatorname{tr}}

\newcommand{\vf}{\varphi}
\newcommand{\Id}{\mathds{1}}

\newcommand{\dist}{\mathrm{dist}}

\newcommand{\supp}{\operatorname{supp}}

\newcommand{\im}{\operatorname{im}}

\usepackage{environ}

\NewEnviron{malign}{%
	\begin{align}\begin{split}
			\BODY
	\end{split}\end{align}
}

\newcommand{\abs}[1]{\left\lvert#1\right\rvert}

\newcommand{\eq}[1]{\begin{align*}#1\end{align*}}

\newcommand{\eql}[1]{\begin{align}#1\end{align}}

\newcommand{\br}[1]{\left(#1\right)}

\title{Instantons in a Double-Well are Poisson Distributed}
\author{	\href{mailto:jacobshapiro@princeton.edu}{Jacob Shapiro}\\
	{\footnotesize Department of Mathematics, Princeton University}
}

\begin{document}
\reqnomode
\maketitle
\begin{abstract}
We give a rigorous realization of the dilute instanton picture for a
semiclassical Schrödinger operator with a symmetric double-well potential
on $\RR^n$. Using a localized Feynman--Kac representation, we decompose the
heat-kernel trace according to the number of passages made by a Brownian
bridge between shrinking neighborhoods of the two wells. We identify the
weight of one passage with a hopping coefficient $\rho_\lambda$ 
\eq{
\rho_\lambda
=\int_{\partial\Omega}
\left(
\br{\nabla\overline{\vf_{\lambda,0}^\Omega}}
\vf_{\lambda,0}^{-\Omega}
-
\overline{\vf_{\lambda,0}^\Omega}
\nabla\vf_{\lambda,0}^{-\Omega}
\right)\cdot\nu\,.
}
On the exponentially long time scale $\beta=N/\abs{\rho_\lambda}$, the number of
passages converges, for every fixed $N>0$, to a Poisson random variable of
mean $N$. We identify $-\frac{1}{\lambda}\log\abs{\rho_\lambda}\to S(d,-d)$
and obtain
\eq{
E_1(\lambda)-E_0(\lambda)
=
2\abs{\rho_\lambda}\br{1+o(1)}.
}
Thus the familiar instanton expansion of the double-well eigenvalue
splitting emerges directly from a factorization of the heat-kernel trace.
\end{abstract}

\section{Introduction}

Tunneling in a symmetric double-well potential is one of the basic
semiclassical mechanisms by which a classically degenerate energy level
splits into two nearby quantum energy levels. For the Schrödinger operator
\eq{
H_\lambda=-\Delta+\lambda^2V(X),
}
the two lowest eigenstates are concentrated near the two minima of $V$,
while their symmetric and antisymmetric combinations have eigenvalues
$E_0(\lambda)$ and $E_1(\lambda)$ separated by an exponentially small
quantity. The exponential scale of this splitting and, under stronger
hypotheses, its leading prefactor are classical subjects of semiclassical
analysis; see, for example,
\cite{Simon_1984_10.2307/2007072,Helffer_Sjostrand_1984}.

One of the physical descriptions is given by instantons. In
Euclidean time, an instanton is a classical zero-energy trajectory joining
the two minima. The standard dilute-gas picture treats repeated transitions
between the wells as approximately independent events. If $\abs{\rho_\lambda}$
denotes the weight of one such transition, the contribution of trajectories
containing exactly $k$ transitions over a time interval of length $\beta$
is then expected to be proportional to
\eq{
\frac{\br{\beta\abs{\rho_\lambda}}^k}{k!}.
}
Summing over even and odd values of $k$ produces the hyperbolic functions
which encode the symmetric and antisymmetric eigenvalues. This argument is
fundamental in the physics literature
\cite{Coleman1985Uses,VainshteinEtAl1982ABC} as it is a stepping stone to instantons in gauge theories.

The ground-state-diffusion and hitting-time approach to multiple tunneling goes back to Jona-Lasinio, Martinelli, and Scoppola \cite{JonaLasinioMartinelliScoppola1981}, while Combes, Duclos, and Seiler \cite{CombesDuclosSeiler1983} obtained convergent expansions in a tunneling parameter. Related dilute-interface asymptotics were developed for one-dimensional Kac and Allen--Cahn models in \cite{COP93,BBB08,Web10,OWW14}, while process-level convergence to finite-state metastable Markov chains in a rather different small-noise setting was proved in \cite{RS23}. Most directly related to the present work, Bertini, Butt{\`a}, and Di Ges{\`u} \cite{BBDG25} proved, for the scalar quartic $\phi^4_1$ model, sharp factorial asymptotics for the number of transition layers and convergence of their locations to a Poisson point process on the natural exponentially long scale. For general symmetric double-well potentials on $\mathbb{R}^n$, we identify the Poisson intensity with the semiclassical hopping coefficient (see \cite{FeffermanLeeThorpWeinstein2018,FeffermanShapiroWeinstein2022,Fefferman2025}, also called the interaction matrix element \cite{Helffer_Sjostrand_1984,DimassiSjostrand1999InteractionMatrix}) and hence simultaneously with its flux formula, exponential action, and the leading spectral splitting. We prove that this hopping coefficient is precisely the limiting weight of the one-passage sector, that on the time scale $\beta=N/\abs{\rho_\lambda}$ the number of geometrically defined passage events converges to a Poisson law of mean $N$, and that
\[
    \lim_{\lambda\to\infty}
    -\frac{1}{\lambda}\log\abs{\rho_\lambda}
    =
    S(d,-d).
\]
Moreover, the leading eigenvalue splitting satisfies
\[
    E_1(\lambda)-E_0(\lambda)
    =
    2\abs{\rho_\lambda}(1+o(1))\,.
\]

There are several difficulties in turning this picture into a direct
argument about the heat kernel. A classical instanton approaches either
minimum only at infinite Euclidean time, so an individual transition does
not occupy a canonically specified finite interval. Brownian trajectories
also need not consist of visibly isolated classical instantons: successive
passages may occur close together and form clusters. Finally, the natural
observation time is of order $\abs{\rho_\lambda}^{-1}$ and is therefore
exponentially large in $\lambda$. It is consequently not enough to obtain
an asymptotic formula for each fixed number of transitions; one also needs
a summable estimate uniform in that number in order to recover the full
partition function.

We address these issues directly through the Feynman--Kac representation.
We localize the endpoints of the heat kernel to shrinking neighborhoods of
the two wells and use successive hitting times of these neighborhoods to
define the number of passages of a path. This gives an exact decomposition
of the localized heat-kernel traces into partition functions
$Z_{\lambda,k}(\beta)$ with definite number of passages. The hopping coefficient
is defined intrinsically by
\eq{
\rho_\lambda
:=
-\lim_{\beta\to\infty}
\frac{Z_{\lambda,1}(\beta)}
{\beta Z_{\lambda,0}(\beta)}.
}
This choice fixes the normalization before any dilute-gas approximation is
made and agrees with the normalization of \cite{Fefferman2025}.

Our main result is that, for every fixed $N>0$,
\eq{
\frac{
Z_{\lambda,k}\left(\frac{N}{\abs{\rho_\lambda}}\right)
}{
Z_{\lambda,0}\left(\frac{N}{\abs{\rho_\lambda}}\right)
}
\longrightarrow
\frac{N^k}{k!}
\qquad
\left(k\in\NN_{\geq0}\right),
}
together with a summable bound which is uniform in $k$ and
$\lambda$. It follows that the number of passages on the time scale
$N/\abs{\rho_\lambda}$ converges to the Poisson law of mean $N$. The proof
implements the dilute-gas intuition by separating paths into two classes:
those for which all passages are well separated and those containing at
least one cluster. Repeated application of the strong Markov property
factorizes the first class, while logarithmic path estimates show that the
second has vanishing contribution.

The Poisson law immediately determines the ratio between the even- and
odd-passage sectors. Comparing this ratio with the spectral expansion of
the heat kernel yields
\eq{
E_1(\lambda)-E_0(\lambda)
=
2\abs{\rho_\lambda}\br{1+o(1)}.
}
Independently, the one-passage decomposition identifies $\rho_\lambda$ with
the boundary Wronskian flux between the two one-well ground states. Uniform
semiclassical path estimates then give
\eq{
-\frac1\lambda\log\abs{\rho_\lambda}
\longrightarrow
S(d,-d),
}
where $S(d,-d)$ is the least Euclidean action required to travel between
the two minima.

The paper is organized as follows. In \cref{thm:Poisson distribution} and
the surrounding discussion, we introduce the passage sectors, state the
Poisson limit, and deduce the eigenvalue splitting. We then identify the
resulting normalization of $\rho_\lambda$ with the standard hopping
coefficient and determine its exponential scale. The final step is the
factorization of the $k$-passage sectors, including the estimate excluding
clustered passages. The appendix proves the logarithmic heat-kernel,
first-exit, and ground-state estimates needed for the particular
exterior-ball domains used throughout the paper.

\subsection{Setting}
For arbitrary $n\in\NN_{\geq1}$, we define our double-well Hamiltonian on $L^2(\RR^n)$ as \eql{
H_\lambda := -\Delta + \lambda^2 V(X)\qquad(\lambda >0 \text{ large})\,.
} We make the following assumptions about this Hamiltonian:
\begin{enumerate}
    \item $V\in C^2(\RR^n\to[0,\infty))$ has precisely two non-degenerate global minima located at $\pm d$ for some $d\in\RR^n\setminus\Set{0}$, and $V(\pm d)=0$.
    \item If $R$ is the unitary operator implementing reflection, i.e., $R \psi\equiv\psi(-\cdot)$ then $RV = V$.
    \item There exists some $c_\infty>0$ such that \eql{
\liminf_{\norm{x}\to\infty}V\left(x\right)\geq c_{\infty}
    } and $H_\lambda$ has discrete spectrum below $\lambda^2 c_\infty$. Its eigenvalues are denoted \eql{
    E_0(\lambda)\leq E_1(\lambda) \leq \cdots<\lambda^2c_\infty\,.
    } We assume that $E_0$ corresponds to the $+1$ eigenvalue of $R$ and $E_1$ to its $-1$ eigenvalue.

    \item We assume that there exists some $c_{\rm gap}>0$ such that \eql{
    \inf_\lambda \frac{1}{\lambda}\br{E_2(\lambda)-E_1(\lambda) }\geq c_{\rm gap} 
    }
    \item For the two one-well domains $\Omega_\lambda$ and
    $-\Omega_\lambda$ introduced below, we let $h_\lambda^{\pm
    \Omega_\lambda}$ denote the corresponding Dirichlet restrictions of
    $H_\lambda$. Their eigenstates and eigenvalues below
    $\lambda^2c_\infty$ are denoted by
    $\Set{\vf_{\lambda,j}^{\pm\Omega_\lambda}}_{j=0}^\infty$ and
    $\Set{e_{\lambda,j}^{\pm\Omega_\lambda}}_{j=0}^\infty$. We assume that
    there exist constants $c,C>0$ and some fixed
    $\eta\in(0,\frac12)$ such that, uniformly for
    \eq{
    (D_\lambda,q)
    \in
    \Set{
    (\Omega_\lambda,d),
    (-\Omega_\lambda,-d)
    },
    }
    \eq{
    c\lambda\leq e_{\lambda,0}^{D_\lambda}\leq C\lambda,
    \qquad
    e_{\lambda,1}^{D_\lambda}-e_{\lambda,0}^{D_\lambda}
    \geq c\lambda,
    }
    and
    \eq{
    \sup_{x\in B_{\lambda^{-1/2+\eta}}(q)}
    \left|
    \frac1\lambda\log\vf_{\lambda,0}^{D_\lambda}(x)
    \right|
    \longrightarrow0.
    }
\end{enumerate}

All of these assumptions may be verified if we assume, for example, that $V$ has a good harmonic oscillator approximation near $\pm d$, as is customary in the semiclassical analysis literature \cite{Simon_1984_10.2307/2007072,Helffer_Sjostrand_1984}; see \cref{sec:verification-one-well-assumptions} for completeness.

The relevant classical Euclidean action in this problem is given by \eql{\label{eq:classical Euclidean action}
S_T(\gamma) := \int_0^T\br{\frac{1}{4}\norm{\dot{\gamma}(t)}^2 + V\br{\gamma(t)}}\dif{t}\,.
} In principle its domain is the space of classical paths $ H^{1}\left(\left[0,T\right]\to\mathbb{R}^{n}\right)$ and the classical equation of motion is \eql{\label{eq:classical EoM}
\frac12 \ddot{\gamma} = \br{\nabla V }\circ\gamma\,.
} Moreover, along such classical paths obeying the Newton equation of motion, the energy
\eq{
E(\gamma) := \frac14 \norm{\dot{\gamma}}^2 - V\circ\gamma
} is conserved. 
\begin{defn}[Instanton] An \emph{instanton} is a solution $\gamma:\RR\to\RR^n$ of \cref{eq:classical EoM} with the boundary conditions $\gamma(\pm\infty)=\pm d$ and $\dot{\gamma}(\pm\infty)=0$. 
\end{defn}
We note that such solutions have zero energy. Moreover, the constraints on having zero velocity and having end-points at the minima cannot be satisfied at finite time.

We will mostly be interested in its infimum value as $\gamma$ ranges over paths with prescribed end-points: 
\eql{
S_T(A,B) := \inf_{\gamma\in H^{1}\left(\left[0,T\right]\to\mathbb{R}^{n}\right):\gamma(0) \in A\land \gamma(T) \in B}S_T(\gamma)\qquad(A,B\subseteq \RR^n)
} as well as 
\eql{
S(A,B) := \inf_{T>0} S_T(A,B)\qquad(A,B\subseteq \RR^n)\,.
} With abuse, if $A$ or $B$ are singletons we just use the point instead of the singleton containing the point. Thus $S(x,y)$ is the extremal classical Euclidean action for paths starting at $x$ and ending at $y$.

For an open set $D\subseteq\RR^n$ and $x,y\in\overline D$, we also use the constrained actions
\eql{
S_{T,D}(x,y)
:=
\inf_{\substack{
\gamma\in H^1([0,T];\RR^n)\\
\gamma(0)=x,\ \gamma(T)=y\\
\gamma((0,T))\subseteq D
}}
\int_0^T
\left(
\frac14\norm{\dot\gamma(t)}^2+V\br{\gamma(t)}
\right)\dif{t}
}
and
\eql{
S_D(x,y):=\inf_{T>0}S_{T,D}(x,y)\,.
}
The use of the open interval $(0,T)$ permits either endpoint to lie on
$\partial D$.

Our main result is that as $\lambda\to\infty$, the heat-kernel associated with $H_\lambda$ may be factorized according to instanton "events" which are distributed with a Poisson distribution of intensity $N$. Once this is established the well-known eigenvalue splitting follows as a consequence. 

\paragraph{Acknowledgements}
JS was supported in part by NSF grant DMS-2510207 and received access to OpenAI products through the ChatGPT for Academic Researchers program. We thank Michael Aizenman for stimulating discussions.

\section{The heat kernel trace and the eigenvalue splitting}

The general strategy to obtain the eigenvalue splitting is to compare $\exp(-t H_\lambda)(d,d)$ with $\exp(-t H_\lambda)(d,-d)$:
\eq{
\exp(-t H_\lambda)(x,y) &= \sum_{j\geq0:E_j(\lambda)<\lambda^2c_\infty}\exp(-t E_j(\lambda))P_{\lambda,j}(x,y) + \int_{E\geq \lambda^2c_\infty}\exp(-t E)\dif{Q_{H_\lambda}(E)}(x,y)\qquad(t>0)
} where $P_{\lambda,j}$ is the spectral projection onto the $j$th discrete eigenvalue and $\dif{Q_{H_\lambda}(E)}$ is the projection-valued measure. Hence reflection symmetry implies $P_{\lambda,0}(d,-d)=P_{\lambda,0}(d,d)$ and $P_{\lambda,1}(d,-d)=-P_{\lambda,1}(d,d)$ so that
\eq{
\exp(-t H_\lambda)(d,\pm d) &= \exp(-t E_0(\lambda))P_{\lambda,0}(d,d)\pm\exp(-t E_1(\lambda))P_{\lambda,1}(d,d)+R_{\lambda}^t(d,\pm d)
} so that
\eq{
E_1(\lambda)-E_0(\lambda) = \lim_{t\to\infty}-\frac{1}{t}\log\br{\frac{\exp(-t H_\lambda)(d, d)-\exp(-t H_\lambda)(d,- d)}{\exp(-t H_\lambda)(d, d)+\exp(-t H_\lambda)(d,- d)}}\,.
}

It will actually be more convenient to study $\tr\br{\exp\br{-t H_\lambda}}$ which should anyway get its main "contribution" from $\exp(-t H_\lambda)(d, d)$. The heat kernel, however, may fail to be trace-class (due to the fact $H_\lambda$ may have continuous spectrum above $\lambda^2c_\infty$ which escapes to infinity) and hence we localize it. It is possible in principle to localize to some fixed $O(1)$ box which contains both minima comfortably, but instead we choose a duo of shrinking balls, one around each minimum. Hence let $r_\lambda := \lambda^{-1/4}$ and consider $A_{\lambda,\pm }:=B_{r_\lambda}(\pm d)$. We also have the total neighborhoud $A_\lambda := A_{\lambda,+}\cup A_{\lambda,-}$. We pick $\chi_\lambda\in C^{\infty}\left(\mathbb{R}^{n}\to\left[0,1\right]\right)$ such that $\supp\left(\chi_{\lambda}\right)\subseteq A_{\lambda}$, $\left.\chi_{\lambda}\right|_{B_{\frac{1}{2}r_{\lambda}}\left(d\right)\cup B_{\frac{1}{2}r_{\lambda}}\left(-d\right)}=1$ and $R\chi_{\lambda}=\chi_{\lambda}$. Then
\begin{malign}\label{eq:def of Z and Z_twist}
Z_\lambda(\beta)
&:= \tr(\chi_\lambda(X) \exp\br{-\beta H_\lambda}\chi_\lambda(X))\\
Z_\lambda^\curvearrowright(\beta)
&:= \tr(\chi_\lambda(X) \exp\br{-\beta H_\lambda}R\chi_\lambda(X))\,.
\end{malign} Via Feynman-Kac, these have path integral expansions of the form
\begin{malign}
\label{eq:Feynman-Kac expansion of traces of heat kernels}
Z_\lambda(\beta)
&=
\int_{x\in\mathbb{R}^{n}}
\chi_{\lambda}\left(x\right)^{2}
p_{\beta}\left(x,x\right)
\mathbb{E}_{\gamma\left(0\right)=x,\gamma\left(\beta\right)=x}
\left[
\exp\left(
-\lambda^{2}\int_{0}^{\beta}V\circ\gamma
\right)
\right]\dif{x}
\\
Z_\lambda^\curvearrowright(\beta)
&=
\int_{x\in\mathbb{R}^{n}}
\chi_{\lambda}\left(x\right)^{2}
p_{\beta}\left(x,-x\right)
\mathbb{E}_{\gamma\left(0\right)=x,\gamma\left(\beta\right)=-x}
\left[
\exp\left(
-\lambda^{2}\int_{0}^{\beta}V\circ\gamma
\right)
\right]\dif{x}.
\end{malign} with 
\eql{p_{\beta}\left(x,y\right)\equiv\left(4\pi \beta\right)^{-\frac{n}{2}}\exp\left(-\frac{1}{4\beta}\norm{x-y}^{2}\right)\qquad\left(x,y\in\mathbb{R}^{n},\beta>0\right)\,.} Here $\mathbb{E}_{\gamma\left(0\right)=x,\gamma\left(\beta\right)=y}
\left[\cdot\right]$ is the expectation w.r.t. the Brownian-bridge measure starting at $x$ and ending at $y$ at time $\beta$. The representation \cref{eq:Feynman-Kac expansion of traces of heat kernels} is useful for us over \cref{eq:def of Z and Z_twist} as now we can start conditioning geometrically on the Brownian path $\gamma$.

Indeed, within the context of the Brownian bridge path $\gamma$ appearing in \cref{eq:Feynman-Kac expansion of traces of heat kernels} we define a sequence of stopping times as follows. We know that $\gamma(0)\in A_{\lambda,\pm}$ and so let $\sigma_0$ be the parity in which $\gamma(0)$ starts. Define $\sigma_j:=(-1)^j\sigma_0$. Then $\tau_0:=0$ and
\eql{\tau_{j}\left(\gamma\right)	:=	\inf\left(\Set{t>\tau_{j-1}(\gamma)|\gamma\left(t\right)\in A_{\lambda,\sigma_{j}}}\right)\qquad\left(j\geq1\right)\,.} We then have also the total number of instantons up to time $\beta$ as 
\eql{
\mathcal{N}_{\beta}\left(\gamma\right)	:=	\sup\left(\Set{j|\tau_{j}(\gamma)\leq \beta}\right)\,.
} With these stopping times, we may then define the definite $k$-instanton partition function as 
\eql{
Z_{\lambda,k}(\beta) := \int_{x\in\mathbb{R}^{n}}\chi_{\lambda}\left(x\right)^{2}p_{\beta}\left(x,\left(-1\right)^{k}x\right)\mathbb{E}_{\gamma\left(0\right)=x,\gamma\left(\beta\right)=\left(-1\right)^{k}x}\left[\exp\left(-\lambda^{2}\int_{0}^{\beta}V\circ\gamma\right)\chi_{\Set{\mathcal{N}_{\beta}\left(\gamma\right)=k}}\right]\dif{x}\,.
}

An emergent parameter, which we call \emph{the hopping coefficient}, emerges via
\eq{
\frac{Z_{\lambda,1}(\beta)}{ Z_{\lambda,0}(\beta)} \sim \beta\abs{\rho_\lambda}\qquad(\beta\to\infty)
} i.e., $\abs{\rho_\lambda}$ will be the "weight" of one instanton event, and a key factorization result will exhibit then
\eq{
\frac{Z_{\lambda,k}(\beta)}{ Z_{\lambda,0}(\beta)} \sim \frac{\br{\beta\abs{\rho_\lambda}}^k}{k!}\,.
} 

This prompts us to define a probability measure on $\NN_{\geq0}$ given by
\eql{
\mathbb{P}_{\lambda,\beta}^{\text{instanton}}\left[\Set{k}\right] := \frac{Z_{\lambda,k}\left(\beta\right)}{Z_{\lambda}\left(\beta\right)+Z_{\lambda}^\curvearrowright \left(\beta\right)} = \frac{Z_{\lambda,k}\left(\beta\right)}{\sum_{\ell=0}^\infty Z_{\lambda,\ell}\left(\beta\right)}\,.
} A natural time-scale to pick then is $\beta_\lambda := \frac{N}{\abs{\rho_\lambda}}$ for some fixed $N>0$, which should make $\mathbb{P}_{\lambda,\beta_\lambda}^{\text{instanton}}$ converge, as $\lambda\to\infty$, to a probability distribution whose mean number of instantons is $N$.

Our main result is thus
\begin{thm}\label{thm:Poisson distribution}
    There exists some $\rho_{\lambda}<0$, $\log(\lambda)\abs{\rho_\lambda}\to0$, namely,
    \eql{\label{eq:hopping coefficient in terms of partition functions}
    \rho_\lambda := -\lim_{\beta\to\infty}\frac{Z_{\lambda,1}(\beta)}{\beta Z_{\lambda,0}(\beta)}
    }
    such that for any $N>0$ (independent of $\lambda$), \eql{\lim_{\lambda\to\infty}\mathbb{P}_{\lambda,\frac{N}{\abs{\rho_{\lambda}}}}^{\text{instanton}}\left[\Set{k}\right]	=	\ee^{-N}\frac{N^{k}}{k!}\qquad\left(k\in\mathbb{N}_{\geq0}\right)} so that instantons are distributed according to a Poisson distribution.
\end{thm}

Now we connect the eigenvalue splitting with the emergent quantity $\rho_\lambda$. The identification relies on the above theorem; once we explicitly define $\rho_\lambda$, we shall identify it as the hopping coefficient (see e.g. \cite{FeffermanShapiroWeinstein2022,Fefferman2025}) and thus there is an entirely independent route (via the Schur complement) to associate it with the eigenvalue splitting. The proof of the corollary below is yet another alternative.
\begin{cor}\label{cor:eigenvalue splitting controlled by rho}
    For the $\rho_\lambda$ defined above, we have
    \eql{
    \lim_{\lambda\to\infty}\frac{E_1(\lambda)-E_0(\lambda)}{2\abs{\rho_\lambda}} = 1\,.
    }
\end{cor}

\begin{proof}
The parity of the endpoint implies the exact decompositions
\eq{
Z_\lambda(\beta)
&=
\sum_{m=0}^{\infty}Z_{\lambda,2m}(\beta),\\
Z_\lambda^{\curvearrowright}(\beta)
&=
\sum_{m=0}^{\infty}Z_{\lambda,2m+1}(\beta).
}
Hence, upon setting
\eq{
\beta_{\lambda,N}:=\frac{N}{2\abs{\rho_\lambda}},
}
we obtain
\eq{
A_\lambda(N)
&:=
\frac{
Z_\lambda(\beta_{\lambda,N})
-
Z_\lambda^{\curvearrowright}(\beta_{\lambda,N})
}{
Z_\lambda(\beta_{\lambda,N})
+
Z_\lambda^{\curvearrowright}(\beta_{\lambda,N})
}\\
&=
\sum_{k=0}^{\infty}
(-1)^k
\mathbb{P}_{\lambda,\beta_{\lambda,N}}^{\mathrm{instanton}}
\left[\Set{k}\right].
}
Applying the theorem, we have
\eq{
\lim_{\lambda\to\infty}
\mathbb{P}_{\lambda,\beta_{\lambda,N}}^{\mathrm{instanton}}
\left[\Set{k}\right]
=
\ee^{-N/2}\frac{(N/2)^k}{k!}.
}
Since the limiting Poisson masses sum to one, this pointwise
convergence implies convergence in total variation (discrete Scheffé lemma). It follows that
\eq{
\lim_{\lambda\to\infty}A_\lambda(N)
&=
\sum_{k=0}^{\infty}
(-1)^k\ee^{-N/2}\frac{(N/2)^k}{k!}\\
&=
\ee^{-N/2}
\sum_{k=0}^{\infty}\frac{(-N/2)^k}{k!}
=
\ee^{-N}.
}

Let
\eq{
\Pi_\pm:=\frac{\Id\pm R}{2}
}
be the projections onto the even and odd subspaces, respectively.
Since $H_\lambda$ and $\chi_\lambda$ commute with $R$, we have
\eq{
Z_\lambda(\beta)\pm Z_\lambda^{\curvearrowright}(\beta)
=
2\tr\br{
\chi_\lambda
\exp\br{-\beta H_\lambda}
\Pi_\pm
\chi_\lambda
}.
}
The range of $P_{\lambda,0}$ lies in the even subspace, whereas the
range of $P_{\lambda,1}$ lies in the odd subspace. Set
\eq{
a_{\lambda,j}
:=
\tr\br{
\chi_\lambda P_{\lambda,j}\chi_\lambda
},
\qquad j\in\Set{0,1}\,,\,\qquad \calR_{\lambda,\pm}(\beta) := \tr\br{\chi_\lambda(X)\exp\br{-\beta H_\lambda}\Pi_\pm \chi_{[E_2(\lambda),\infty)}(H_\lambda) \chi_\lambda(X)}\,.
} 
Hence
\eq{
A_\lambda(N) = \frac{a_{\lambda,1}
\exp\br{
-N\frac{E_1(\lambda)}
        {2\abs{\rho_\lambda}}
}+\calR_{\lambda,-}(\beta_{\lambda,N})}{a_{\lambda,0}\exp\br{
-N\frac{E_0(\lambda)}
        {2\abs{\rho_\lambda}}
}+\calR_{\lambda,+}(\beta_{\lambda,N})}
}

Now choose e.g.
$\tau_\lambda:=\frac{1}{\lambda^2}$.
Using the spectral theorem, we obtain
\eq{
\calR_{\lambda,\pm}(\beta)
\leq
\exp\br{
-\br{\beta-\tau_\lambda}E_2(\lambda)
}
\tr\br{
\chi_\lambda
\exp\br{-\tau_\lambda H_\lambda}
\chi_\lambda
}.
}
Dividing by the corresponding leading term gives, for
$j\in\Set{0,1}$,
\eq{
\frac{
\calR_{\lambda,\pm}(\beta)
}{
a_{\lambda,j}\exp\br{-\beta E_j(\lambda)}
}
\leq
\frac{
\exp\br{\tau_\lambda E_2(\lambda)}
\tr\br{
\chi_\lambda
\exp\br{-\tau_\lambda H_\lambda}
\chi_\lambda
}
}{
a_{\lambda,j}
}
\exp\br{
-\beta\br{E_2(\lambda)-E_j(\lambda)}
}.
}
At $\beta=\beta_{\lambda,N}$, the relevant suppression factor satisfies
\eq{
\exp\br{
-\beta_{\lambda,N}
\br{E_2(\lambda)-E_1(\lambda)}
}
\leq
\exp\br{
-\frac{N \lambda c_{\rm gap}}{2\abs{\rho_\lambda}}
}.
}
Since $\log(\lambda)\abs{\rho_\lambda}\to0$, this factor tends to zero. To conclude that
the entire right-hand side tends to zero, however, we must also
control the prefactor. By the Feynman--Kac formula and the non-negativity of $V$,
\eq{
\tr\br{
\chi_\lambda
\exp\br{-H_\lambda/\lambda^2}
\chi_\lambda
}
\leq
\br{\frac{\lambda^2}{4\pi}}^{n/2}
\int_{x\in\RR^n}\chi_\lambda(x)^2\dif{x}.
}
Because $\chi_\lambda$ is supported in two balls of radius
$\lambda^{-1/4}$,
\eq{
\int_{x\in\RR^n}\chi_\lambda(x)^2\dif{x}
=
\Ord{\lambda^{-n/4}},
}
and hence
\eq{
\tr\br{
\chi_\lambda
\exp\br{-H_\lambda/\lambda^2}
\chi_\lambda
}
=
\Ord{\lambda^{3n/4}}.
}
We have $E_2(\lambda)\leq\lambda^2 c_{\infty}$ and
$a_{\lambda,j}\gtrsim1-\calO(\lambda^{-1/2})$ by basic energy estimates. The prefactor is therefore at most
polynomial in $\lambda$. Thus, for some $C,M>0$,
\eq{
\frac{
\calR_{\lambda,\pm}(\beta_{\lambda,N})
}{
a_{\lambda,j}
\exp\br{-\beta_{\lambda,N}E_j(\lambda)}
}
\leq
C\lambda^M
\exp\br{
-\frac{N \lambda c_{\rm gap}}{2\abs{\rho_\lambda}}
}
\longrightarrow0\,.
}

We conclude that for every fixed
$N>0$,
\eq{
A_\lambda(N)
=
\frac{a_{\lambda,1}}{a_{\lambda,0}}
\exp\br{
-N\frac{E_1(\lambda)-E_0(\lambda)}
        {2\abs{\rho_\lambda}}
}
\br{1+o(1)}.
}

Applying this asymptotic with $N=1$ and $N=2$ and taking
their ratio eliminates the localization factors:
\eq{
\frac{A_\lambda(2)}{A_\lambda(1)}
=
\exp\br{
-\frac{E_1(\lambda)-E_0(\lambda)}
       {2\abs{\rho_\lambda}}
}
\br{1+o(1)}.
}
On the other hand, the Poisson limit gives
\eq{
\lim_{\lambda\to\infty}
\frac{A_\lambda(2)}{A_\lambda(1)}
=
\frac{\ee^{-2}}{\ee^{-1}}
=
\ee^{-1}.
}
Taking logarithms therefore yields
\eq{
\lim_{\lambda\to\infty}
\frac{E_1(\lambda)-E_0(\lambda)}
     {2\abs{\rho_\lambda}}
=
1\,.
}
\end{proof}
\begin{lem}
    The limit in the definition of $\rho_\lambda$ from \cref{eq:hopping coefficient in terms of partition functions} exists and equals
    \eql{\label{eq:flux equation for hopping coeffivient}
\rho_\lambda
=
\int_{y\in\partial\Omega_\lambda}
\nu\cdot\br{\nabla\vf_{\lambda,0}^{\Omega_\lambda}}
\vf_{\lambda,0}^{-\Omega_\lambda}
\dif{y}
} where $\Omega_\lambda := \RR^n\setminus A_{\lambda,-}$; when
$n=1$, $\Omega_\lambda$ denotes the connected component containing
$d$; correspondingly, $-\Omega_\lambda$ denotes its reflection.  $\nu:\partial \Omega_\lambda \to \RR^n$ is the unit normal pointing out of $\Omega_\lambda$.
\end{lem}
\begin{proof}
Let $\nu_\lambda$ denote the outward unit normal to
$\Omega_\lambda$ on $\partial\Omega_\lambda$, and set
\eq{
-\Omega_\lambda
&:=
\Set{-x\mid x\in\Omega_\lambda}
=
\RR^n\setminus A_{\lambda,+},\\
\chi_{\lambda,\pm}
&:=
\chi_\lambda\chi_{A_{\lambda,\pm}}.
}
By reflection symmetry, the Dirichlet operators
$h_\lambda^{\Omega_\lambda}$ and $h_\lambda^{-\Omega_\lambda}$ are
unitarily equivalent. Thus their ground-state energies coincide; write
\eq{
e_\lambda^D
:=
e_{\lambda,0}^{\Omega_\lambda}
=
e_{\lambda,0}^{-\Omega_\lambda}.
}
Their positive normalized ground states may be chosen so that
\eq{
R\vf_{\lambda,0}^{\Omega_\lambda}
=
\vf_{\lambda,0}^{-\Omega_\lambda},
\qquad
\vf_{\lambda,0}^{-\Omega_\lambda}(-x)
=
\vf_{\lambda,0}^{\Omega_\lambda}(x).
}
Set
\eq{
a_\lambda
:=
\norm{
\chi_{\lambda,+}
\vf_{\lambda,0}^{\Omega_\lambda}
}^2
=
\norm{
\chi_{\lambda,-}
\vf_{\lambda,0}^{-\Omega_\lambda}
}^2.
}

A path contributing to $Z_{\lambda,0}(\beta)$ is killed when it enters
the opposite well. Consequently, we may replace $H_\lambda$ with $h_\lambda^{\Omega_\lambda}$ in the heat kernel to get
\eq{
Z_{\lambda,0}(\beta)
&=
\tr\br{
\chi_{\lambda,+}
\exp\br{-\beta h_\lambda^{\Omega_\lambda}}
\chi_{\lambda,+}
}+
\tr\br{
\chi_{\lambda,-}
\exp\br{-\beta h_\lambda^{-\Omega_\lambda}}
\chi_{\lambda,-}
} \\ 
&= 2\tr\br{
\chi_{\lambda,+}
\exp\br{-\beta h_\lambda^{\Omega_\lambda}}
\chi_{\lambda,+}
}
} where the last equality is by reflection symmetry.

Let
\eq{
P_{\lambda,0}^{\Omega_\lambda}
&:=
\vf_{\lambda,0}^{\Omega_\lambda}
\otimes
\br{\vf_{\lambda,0}^{\Omega_\lambda}}^*,
\qquad
Q_{\lambda}^{\Omega_\lambda}
:=
\Id-P_{\lambda,0}^{\Omega_\lambda},\\
\delta_\lambda^D
&:=
e_{\lambda,1}^{\Omega_\lambda}
-
e_{\lambda,0}^{\Omega_\lambda}.
}
By the assumed one-well spectral-gap estimate, there exists
$c_{\rm gap}^D>0$ such that
\eq{
\delta_\lambda^D
\geq
c_{\rm gap}^D\lambda
}
for all sufficiently large $\lambda$.
Resolving 
$\Id=P_{\lambda,0}^{\Omega_\lambda}
+Q_{\lambda}^{\Omega_\lambda}$, we obtain the exact decomposition
\eql{\label{eq:zero passage spectral decomposition}
Z_{\lambda,0}(\beta)
=
2a_\lambda\ee^{-\beta e_\lambda^D}
+
2\calR_{\lambda,0}(\beta),
}
where
\eq{
\calR_{\lambda,0}(\beta)
:=
\tr\br{
\chi_{\lambda,+}
\exp\br{-\beta h_\lambda^{\Omega_\lambda}}
Q_{\lambda}^{\Omega_\lambda}
\chi_{\lambda,+}
}.
}

To estimate the remainder, fix
$\tau_\lambda:=\lambda^{-1}$ and suppose that
$\beta>\tau_\lambda$. Since
$Q_{\lambda}^{\Omega_\lambda}$ commutes with
$h_\lambda^{\Omega_\lambda}$ and the spectrum of
$h_\lambda^{\Omega_\lambda}$ on
$\im Q_{\lambda}^{\Omega_\lambda}$ is contained in
$[e_{\lambda,1}^{\Omega_\lambda},\infty)$, the spectral theorem gives
\eq{
\calR_{\lambda,0}(\beta)
&\leq
\ee^{-\br{\beta-\tau_\lambda}
e_{\lambda,1}^{\Omega_\lambda}}
\tr\br{
\chi_{\lambda,+}
\exp\br{-\tau_\lambda h_\lambda^{\Omega_\lambda}}
\chi_{\lambda,+}
}.
}
Consequently,
\eq{
\frac{
\calR_{\lambda,0}(\beta)
}{
a_\lambda\ee^{-\beta e_\lambda^D}
}
&\leq
C_\lambda
\ee^{-\beta\delta_\lambda^D}\\
&\leq
C_\lambda
\ee^{-c_{\rm gap}^D\lambda\beta},
}
where
\eq{
C_\lambda
:=
\frac{
\ee^{\tau_\lambda e_{\lambda,1}^{\Omega_\lambda}}
}{
a_\lambda
}
\tr\br{
\chi_{\lambda,+}
\exp\br{-\tau_\lambda h_\lambda^{\Omega_\lambda}}
\chi_{\lambda,+}
}
<
\infty.
}
Moreover, the one-well localization assumption gives
$a_\lambda\geq\exp\br{-o(\lambda)}$, while Feynman--Kac domination by
the free heat kernel and
$e_{\lambda,1}^{\Omega_\lambda}<\lambda^2c_\infty$ give
$C_\lambda\leq\exp\br{O(\lambda)}$. Consequently, using
\cref{eq:hopping-exponent}, for every fixed $N>0$,
\eq{
C_\lambda
\ee^{-c_{\rm gap}^D\lambda N/\abs{\rho_\lambda}}
\longrightarrow0,
}
so the relative remainder in \cref{eq:zero passage flux asymptotic}
is $o(1)$ at $\beta=N/\abs{\rho_\lambda}$.
Thus, as $\beta\to\infty$ with $\lambda$ fixed,
\eql{\label{eq:zero passage flux asymptotic}
Z_{\lambda,0}(\beta)
=
2a_\lambda\ee^{-\beta e_\lambda^D}
\br{
1+
\calO_\lambda\br{
\ee^{-c_{\rm gap}^D\lambda\beta}
}
}.
}

We next decompose a one-passage path at its first entrance into the
opposite well. Denote the Dirichlet heat kernels by
\eq{
K_{\lambda,t}^{\Omega_\lambda}(x,y)
&:=
\exp\br{-t h_\lambda^{\Omega_\lambda}}(x,y),\\
K_{\lambda,t}^{-\Omega_\lambda}(x,y)
&:=
\exp\br{-t h_\lambda^{-\Omega_\lambda}}(x,y).
}
The strong Markov property and the Dirichlet heat-kernel flux formula
give
\begin{malign}
\label{eq:one passage flux convolution}
Z_{\lambda,1}(\beta)
&=
2\int_{x\in A_{\lambda,+}}
\chi_{\lambda,+}(x)^2
\int_{y\in\partial\Omega_\lambda}
\int_{t=0}^\beta
\br{
-\nu_\lambda(y)\cdot
\nabla_yK_{\lambda,t}^{\Omega_\lambda}(x,y)
}
\\
&\qquad\qquad\times
K_{\lambda,\beta-t}^{-\Omega_\lambda}(y,-x)
\dif{t}\dif{y}\dif{x}.
\end{malign}
The factor $2$ accounts for the reflected passage from
$A_{\lambda,-}$ to $A_{\lambda,+}$.

Insert the spectral decompositions of the two heat kernels into
\cref{eq:one passage flux convolution}. The product of their
ground-state contributions is
\eq{
\ee^{-t e_\lambda^D}
\ee^{-(\beta-t)e_\lambda^D}
=
\ee^{-\beta e_\lambda^D},
}
which is independent of $t$. Its integral over $[0,\beta]$ therefore
produces a factor $\beta$. Reflection symmetry also gives
\eq{
\int_{x\in A_{\lambda,+}}
\chi_{\lambda,+}(x)^2
\vf_{\lambda,0}^{\Omega_\lambda}(x)
\vf_{\lambda,0}^{-\Omega_\lambda}(-x)
\dif{x}
=
a_\lambda.
}

The corresponding ground-state contribution to
\cref{eq:one passage flux convolution} is $2a_\lambda\beta\ee^{-\beta e_\lambda^D} q_\lambda
$ where \eq{q_\lambda\equiv
-\int_{y\in\partial\Omega_\lambda}
\nu_\lambda(y)\cdot
\br{\nabla\vf_{\lambda,0}^{\Omega_\lambda}}(y)
\vf_{\lambda,0}^{-\Omega_\lambda}(y)
\dif{y}} is precisely minus the RHS of \cref{eq:flux equation for hopping coeffivient}.
Every remaining contribution contains the orthogonal complement of the
ground state in at least one of the two heat kernels. The one-well
spectral gap gives
\eq{
\norm{
\exp\br{-t h_\lambda^{\Omega_\lambda}}
Q_\lambda^{\Omega_\lambda}
}
\leq
\ee^{-t\br{e_\lambda^D+\delta_\lambda^D}},
}
and, by reflection symmetry, the analogous estimate holds on
$-\Omega_\lambda$. Moreover, $\chi_{\lambda,+}$ is smooth and compactly
supported at positive distance from $\partial\Omega_\lambda$. Splitting
the $t$-integral into the two endpoint regions and the intervening
region, the off-diagonal boundary heat-kernel estimates control the
endpoint regions, while the preceding spectral-gap estimate controls
the intervening region. Hence, for fixed $\lambda$, the sum of all
terms containing at least one orthogonal factor is
$\calO_\lambda\br{\ee^{-\beta e_\lambda^D}}$. We present the full details of these error estimates in \cref{app:one-passage-error-estimates} below.

We have therefore proved
\eql{\label{eq:one passage flux asymptotic}
Z_{\lambda,1}(\beta)
=
2a_\lambda\ee^{-\beta e_\lambda^D}
\br{
\beta q_\lambda
+\calO_\lambda(1)
}.
}

Dividing \cref{eq:one passage flux asymptotic} by
$\beta$ times \cref{eq:zero passage flux asymptotic} yields
\eq{
\rho_\lambda\equiv -\lim_{\beta\to\infty}
\frac{Z_{\lambda,1}(\beta)}
{\beta Z_{\lambda,0}(\beta)}
= -q_\lambda
}
Thus the limit in
\cref{eq:hopping coefficient in terms of partition functions} exists
and equals the asserted boundary flux.
Finally, the ground states are strictly positive on the component of $\pm\Omega_\lambda$ containing $\pm d$ and the Hopf
boundary-point lemma gives
\eq{
-\nu_\lambda\cdot
\nabla\vf_{\lambda,0}^{\Omega_\lambda}
>
0
\qquad\text{on }\partial\Omega_\lambda,
}
so $\rho_\lambda<0$.
\end{proof}

\section{Logarithmic asymptotics for the hopping coefficient}

Although much of this section could be carried out for a general domain
$\Omega$, for simplicity we state and prove only the results needed for
the particular domains $\Omega_\lambda$ and $-\Omega_\lambda$ chosen in
the preceding section. When $n=1$, $\Omega_\lambda$ below means the
connected component of $\RR\setminus A_{\lambda,-}$ containing $d$, and
$-\Omega_\lambda$ means its reflection. This convention does not alter any
of the stopped paths in the preceding section.

Since $\vf_{\lambda,0}^{\Omega}\equiv0$ on
$\partial\Omega$ we see that $\rho_\lambda$ from the previous section agrees with the more general definition of the hopping coefficient, appearing e.g. in \cite[Eq.~(3.1)]{Fefferman2025} (and see references
therein):
\eql{
\rho_\lambda^\Omega
=\int_{\partial\Omega}
\left(
\br{\nabla\overline{\vf_{\lambda,0}^\Omega}}
\vf_{\lambda,0}^{-\Omega}
-
\overline{\vf_{\lambda,0}^\Omega}
\nabla\vf_{\lambda,0}^{-\Omega}
\right)\cdot\nu
\label{eq:hopping coefficient}
}

In this section, with the special choice of $\Omega=\Omega_\lambda$ from the preceding section, we determine the exponential asymptotics of $\rho_\lambda$.

The appendix proves the logarithmic path estimates needed below directly
for these exterior-ball domains. Under the one-well assumptions above,
\cref{lem:exterior-ball-path-estimates}(ii) gives, for every fixed compact
set $K\subseteq\RR^n$,
\eq{
\sup_{\substack{x\in K\cap D_\lambda\\
\dist(x,\partial D_\lambda)\geq r_\lambda}}
\left|
-\frac1\lambda\log\vf_{\lambda,0}^{D_\lambda}(x)
-S_{D_\lambda}(q,x)
\right|
\longrightarrow0,
}
uniformly for
\eq{
(D_\lambda,q)
\in
\Set{
(\Omega_\lambda,d),
(-\Omega_\lambda,-d)
}.
}
The corresponding upper bound remains valid up to the boundary, and the
same lemma gives the polynomial $L^\infty$ bound on the one-well ground
state. In the proof below we use these estimates through the integrated
first-exit estimate in
\cref{lem:exterior-ball-path-estimates}(iii).

\begin{lem}[Hopping exponent]
\label{lem:hopping-exponent}
We have
\eql{
\lim_{\lambda\to\infty}
-\frac1\lambda\log\abs{\rho_\lambda}
=
S(d,-d).
\label{eq:hopping-exponent}
}
\end{lem}

\begin{proof}
Let $\EE$ be the Wiener expectation used in the appendix and set
\eq{
\EE_z[F(\gamma)]
:=
\EE[F(z+\gamma)],
\qquad
\tau_{\Omega_\lambda}(\gamma)
:=
\inf\Set{t>0\mid\gamma(t)\notin\Omega_\lambda}.
}
For any fixed $L>0$, the parabolic Poisson-kernel representation,
integrated over $[0,L/\lambda]$, gives
\eql{
-\rho_\lambda
&=
\frac{e_{\lambda,0}^{\Omega_\lambda}}
{1-\ee^{-\frac{L}{\lambda}e_{\lambda,0}^{\Omega_\lambda}}}
\int_{z\in\Omega_\lambda}
\EE_z\left[
\exp\left(
-\lambda^2
\int_0^{\tau_{\Omega_\lambda}(\gamma)}
V\br{\gamma(t)}\dif{t}
\right)
\right.
\nonumber\\
&\hspace{10em}\left.
\times
\vf_{\lambda,0}^{-\Omega_\lambda}
\br{\gamma\br{\tau_{\Omega_\lambda}(\gamma)}}
\chi_{\Set{
\tau_{\Omega_\lambda}(\gamma)\leq L/\lambda
}}
\right]
\vf_{\lambda,0}^{\Omega_\lambda}(z)\,\dif{z}.
\label{eq:hopping-coefficient-stopped-expectation}
}
Indeed, after writing the expectation by means of the boundary exit
density, symmetry of the Dirichlet heat kernel gives
\eq{
&\int_{z\in\Omega_\lambda}
\vf_{\lambda,0}^{\Omega_\lambda}(z)
\left[
-\nu_\lambda(y)\cdot
\nabla_yK_{\lambda,t}^{\Omega_\lambda}(z,y)
\right]
\dif z
\\
&\qquad=
\ee^{-t e_{\lambda,0}^{\Omega_\lambda}}
\left[
-\nu_\lambda(y)\cdot
\nabla\vf_{\lambda,0}^{\Omega_\lambda}(y)
\right].
}
Integration in $t$ produces the prefactor in
\cref{eq:hopping-coefficient-stopped-expectation}. Since
$e_{\lambda,0}^{\Omega_\lambda}\asymp\lambda$, this prefactor is
subexponential on the scale $\lambda$.

By \cref{lem:exterior-ball-path-estimates}(iii),
\eq{
-\frac1\lambda\log\abs{\rho_\lambda}
&=
\inf_{\substack{
z\in\Omega_\lambda,\ y\in\partial\Omega_\lambda\\
0<T\leq L
}}
\left\{
S_{\Omega_\lambda}(d,z)
+S_{T,\Omega_\lambda}(z,y)
+S_{-\Omega_\lambda}(y,-d)
\right\}
+o(1).
}
For every $y\in\partial\Omega_\lambda$, the dynamic-programming property
of the action gives
\eq{
\inf_{\substack{z\in\Omega_\lambda\\0<T\leq L}}
\left\{
S_{\Omega_\lambda}(d,z)
+S_{T,\Omega_\lambda}(z,y)
\right\}
=
S_{\Omega_\lambda}(d,y).
}
One inequality follows by concatenation. For the reverse inequality,
split an almost-minimizing path sufficiently close to its terminal point
that the duration of its final segment is at most $L$. Consequently,
\eql{
-\frac1\lambda\log\abs{\rho_\lambda}
=
\inf_{y\in\partial\Omega_\lambda}
\left\{
S_{\Omega_\lambda}(d,y)
+S_{-\Omega_\lambda}(y,-d)
\right\}
+o(1).
\label{eq:hopping-action-through-boundary}
}

It remains to identify the limit of the variational expression. By
concatenation,
\eq{
S(d,-d)
\leq
S_{\Omega_\lambda}(d,y)
+S_{-\Omega_\lambda}(y,-d)
\qquad
(y\in\partial\Omega_\lambda).
}
Conversely, fix $\varepsilon>0$ and choose a path from $d$ to $-d$ with
action at most $S(d,-d)+\varepsilon$. Stop it when it first meets
$\partial\Omega_\lambda$, at a point $y$. Its initial segment is
admissible in $\Omega_\lambda$. The straight segment from $y$ to $-d$
lies in $-\Omega_\lambda$ for all sufficiently large $\lambda$ and,
because $V(x)\leq C\norm{x+d}^2$ near $-d$, has action at most
$Cr_\lambda^2$. Hence
\eq{
\inf_{y\in\partial\Omega_\lambda}
\left\{
S_{\Omega_\lambda}(d,y)
+S_{-\Omega_\lambda}(y,-d)
\right\}
\leq
S(d,-d)+\varepsilon+Cr_\lambda^2.
}
Letting first $\lambda\to\infty$ and then
$\varepsilon\downarrow0$ proves the result.
\end{proof}
\subsection{Relation to the FSW hopping coefficient}
In the context where both are defined, we ask to what extent the
Dirichlet hopping coefficient used here agrees with the full-space
single-well overlap coefficient used in the FSW framework
\cite{FeffermanShapiroWeinstein2022,Fefferman2025}. Strictly speaking,
the cited FSW framework concerns magnetic Hamiltonians, whereas here
$\rho_\lambda^{\rm FSW}$ denotes the corresponding non-magnetic
coefficient. Up to a sign and translation convention, the same
non-magnetic full-space overlap integral appears already in
\cite[Eq.~(4.7)]{FeffermanLeeThorpWeinstein2018}, where it governs the
nearest-neighbor coefficient in a periodic tight-binding reduction
rather than the splitting of an isolated double well.

The Dirichlet coefficient used here, on the other hand, is directly
modeled on the classical interaction-matrix construction for
semiclassical multiple wells
\cite{Helffer_Sjostrand_1984,DimassiSjostrand1999InteractionMatrix}.
There the one-well states are ground states of suitable Dirichlet
realizations obtained by excising the other wells, and their Wronskian
flux across a separating hypersurface gives the leading off-diagonal
interaction-matrix element. Thus the Dirichlet coefficient used here
and the full-space FSW coefficient are two realizations of the same
underlying interaction, and the purpose of the present subsection is
to compare them. Related full-space orbital matrix elements also occur
in the earlier LCAO tight-binding literature. We retain the superscript
${\rm FSW}$ merely as a local label distinguishing the full-space
single-well coefficient from the Dirichlet coefficient $\rho_\lambda$
used above.

To relate the two contexts, let us take
$V^{\rm FSW}(x) := v(x-d)+v(x+d)$ for all $x\in\RR^n$ where $v$ is compactly supported and $v,d$ are chosen so that the two supports do not overlap. Moreover, we assume $v:\RR^n\to[-1,0]$, $v=v\br{-\cdot}$ and $\supp(v)\subseteq B_a(0)$ for some $a<\norm{d}$. If we further assume that $v\in C^2$ and it has a unique non-degenerate minimum at the origin with $v(0)=-1$, i.e.,\eq{
v(x) = -1 + \frac12 \ip{x}{\operatorname{Hess}\br{v}_0 x} + o(\norm{x}^2)\qquad(x\in\RR^n)
} then the two contexts overlap, except that there is an overall discrepancy by a constant, which however, does not change the dynamics. Indeed, in the FSW setting, the minimum of the potential is at $-\lambda^2$ whereas here it is at $0$, hence \eq{
V^{\rm FSW} = V-1\,.
}

The single-well Hamiltonian defined by FSW is \eq{
h_\lambda^{\rm FSW} := -\Delta + \lambda^2 v(X)
} with eigenvalues $e_{\lambda,0},e_{\lambda,1},\cdots$ and eigenfunctions $\vf_{\lambda,0},\vf_{\lambda,1},\cdots$. Then the hopping coefficient is defined by FSW as 
\eq{
\rho_\lambda^{\rm FSW} := \ip{T^d \vf_{\lambda,0}}{\br{H_\lambda^{\rm FSW}-e_{\lambda,0}\Id}T^{-d}\vf_{\lambda,0}}
} where $\br{T^y f}(x)\equiv f(x-y)$ is the translation operator. Then in \cite{Fefferman2025} it was shown that if $\Omega\subseteq\RR^n$ has a smooth boundary such that $\supp(v(\cdot-d))\subseteq\operatorname{interior}\br{\Omega}$, $\supp(v(\cdot+d))\subseteq\operatorname{interior}\br{\Omega^c}$ then
\eql{\label{eq:FSW hopping coefficient}
\rho_\lambda^{\rm FSW} = \int_{\partial\Omega}
\left(
\br{\nabla\overline{T^d\vf_{\lambda,0}}}
T^{-d}\vf_{\lambda,0}
-
\overline{T^d\vf_{\lambda,0}}
\nabla T^{-d}\vf_{\lambda,0}
\right)\cdot\nu\,.
} Hence this formula appears identical to \cref{eq:hopping coefficient} but its interpretation is different: in \cref{eq:hopping coefficient} the eigenfunctions were Dirichlet ground states and $\Omega\equiv\Omega_\lambda$ was shrinking on the minima. In \cref{eq:FSW hopping coefficient} the eigenfunctions are translates of the single-well Hamiltonian and $\Omega$ is any set (in particular independent of $\lambda$). 

\begin{prop}
    The two agree asymptotically as $\lambda\to\infty$:
    \eql{
        \lim_{\lambda\to\infty} \frac{\rho_\lambda}{\rho_\lambda^{\rm FSW}} = 1\,.
    }

    In particular combined with the preceding subsection this shows that 
    \eq{
        \lim_{\lambda\to\infty}-\frac{1}{\lambda}\log\br{\abs{\rho_\lambda^{\rm FSW}}} = S(d,-d) = \inf_{T>0} \inf_{\substack{
\gamma\in H^1([0,T];\RR^n)\\
\gamma(0)=d,\ \gamma(T)=-d
}}
\int_0^T
\left(
\frac14\norm{\dot\gamma}^2+V^{\rm FSW}\circ\gamma
+1\right)\,.
    }
\end{prop}
\begin{proof}
The simplest proof goes as follows. Both quantities asymptotically obey 
\eq{
E_1(\lambda)-E_0(\lambda) = 2\abs{\rho_\lambda^\sharp}\br{1+o(1)}\qquad(\sharp = {\rm FSW }\text{ or nothing})\,.
} For $\sharp=\rm FSW$ this was proven in \cite{FeffermanShapiroWeinstein2022} (which remains true in the non-magnetic setting) and otherwise it is a consequence of \cref{cor:eigenvalue splitting controlled by rho}. We avoid giving a direct proof that does not use this fact for brevity, though of course there are various different such direct proofs.

\end{proof}

\section{Factorizing the heat kernel trace}
\label{sec:factorizing-heat-kernel-trace}

The last step in proving \cref{thm:Poisson distribution} is factorizing \eq{
\mathbb{P}_{\lambda,\beta}^{\text{instanton}}\left[\Set{k}\right] \equiv \frac{Z_{\lambda,k}\left(\beta\right)}{Z_{\lambda}\left(\beta\right)+Z_{\lambda}^\curvearrowright \left(\beta\right)}\,.
} What we mean is we imagine that well-separated instanton events should be independent. Hence, roughly speaking we want to establish
\eq{
\mathbb{P}_{\lambda,\beta}^{\text{instanton}}\left[\Set{k}\right] \sim \frac{1}{k!}\mathbb{P}_{\lambda,\beta}^{\text{instanton}}\left[\Set{0}\right] \br{\frac{\mathbb{P}_{\lambda,\beta}^{\text{instanton}}\left[\Set{1}\right]}{\mathbb{P}_{\lambda,\beta}^{\text{instanton}}\left[\Set{0}\right]}}^k\,;
} to that end we first establish
\eq{
\frac{Z_{\lambda,k}\left(\beta\right)}{Z_{\lambda,0}\left(\beta\right)} \sim \frac{1}{k!}\br{\frac{Z_{\lambda,1}\left(\beta\right)}{Z_{\lambda,0}\left(\beta\right)}}^k\,.
}

We now make this precise on the time scale $\beta=N/\abs{\rho_\lambda}$. We shall prove that, for every fixed $N>0$ and $k\in\NN_{\geq0}$,
\eql{
\frac{
Z_{\lambda,k}\br{\frac{N}{\abs{\rho_\lambda}}}
}{
Z_{\lambda,0}\br{\frac{N}{\abs{\rho_\lambda}}}
}
\longrightarrow
\frac{N^k}{k!}.
\label{eq:fixed-k-factorization}
}
In addition, we shall prove that there exists $C_N<\infty$ such that, for all sufficiently large $\lambda$,
\eql{
\frac{
Z_{\lambda,k}\br{\frac{N}{\abs{\rho_\lambda}}}
}{
Z_{\lambda,0}\br{\frac{N}{\abs{\rho_\lambda}}}
}
\leq
C_N2^{-k}
\qquad
(k\in\NN_{\geq0}).
\label{eq:summable-uniform-k-bound}
}
The second estimate is used only to justify summation over $k$.

For $\sigma\in\Set{+1,-1}$, set
\eq{
\Omega_{\lambda,+1}:=\Omega_\lambda,
\qquad
\Omega_{\lambda,-1}:=-\Omega_\lambda,
}
and let $\nu_{\lambda,\sigma}$ denote the outward unit normal to $\Omega_{\lambda,\sigma}$. For
\eq{
x\in\Omega_{\lambda,\sigma},
\qquad
y\in\partial\Omega_{\lambda,\sigma},
}
write
\eql{
J_{\lambda,\sigma}(t;x,y)
:=-\nu_{\lambda,\sigma}(y)\cdot
\nabla_yK_{\lambda,t}^{\Omega_{\lambda,\sigma}}(x,y).
\label{eq:passage-flux-kernel}
}
Thus $J_{\lambda,\sigma}$ is the Feynman--Kac weighted first-exit density from $\Omega_{\lambda,\sigma}$.

By \cref{eq:flux equation for hopping coeffivient} and reflection symmetry,
\eql{
\rho_\lambda
=
\int_{y\in\partial\Omega_{\lambda,\sigma}}
\left[
\nu_{\lambda,\sigma}(y)\cdot
\nabla\vf_{\lambda,0}^{\Omega_{\lambda,\sigma}}(y)
\right]
\vf_{\lambda,0}^{\Omega_{\lambda,-\sigma}}(y)
\dif{y}
\qquad
(\sigma\in\Set{+1,-1}).
\label{eq:rho-both-directions}
}
It is convenient to normalize the corresponding density by setting
\eql{
m_{\lambda,\sigma}(y)
:=
\frac{
\left[
-\nu_{\lambda,\sigma}(y)\cdot
\nabla\vf_{\lambda,0}^{\Omega_{\lambda,\sigma}}(y)
\right]
\vf_{\lambda,0}^{\Omega_{\lambda,-\sigma}}(y)
}{
\abs{\rho_\lambda}
}
\qquad
(y\in\partial\Omega_{\lambda,\sigma}).
\label{eq:equilibrium-passage-density}
}
Then
\eq{
\int_{y\in\partial\Omega_{\lambda,\sigma}}
m_{\lambda,\sigma}(y)\dif{y}=1.
}
For
\eq{
x\in A_{\lambda,\sigma}\cup\partial\Omega_{\lambda,-\sigma},
\qquad
y\in\partial\Omega_{\lambda,\sigma},
}
define the ground-state-normalized passage kernel
\eql{
\mathsf T_{\lambda,\sigma}(t;x,y)
:=
\ee^{t e_\lambda^D}
\frac{
\vf_{\lambda,0}^{\Omega_{\lambda,-\sigma}}(y)
}{
\vf_{\lambda,0}^{\Omega_{\lambda,\sigma}}(x)
}
J_{\lambda,\sigma}(t;x,y).
\label{eq:normalized-passage-kernel}
}
Its ground-state contribution is exactly $\abs{\rho_\lambda} m_{\lambda,\sigma}(y)$, independently of $t$ and of the incoming point $x$.

\subsection{The defect kernel}

We separate a relaxed passage from the part which either occurs too early or has not yet relaxed to the one-well ground state.

\begin{lem}[Integrated defect estimate]
\label{lem:integrated-defect-estimate}
There exists a sequence
\eql{
\varepsilon_\lambda\downarrow0,
\qquad
\lambda\varepsilon_\lambda\longrightarrow\infty,
\label{eq:choice-relaxation-time}
}
such that the following holds. Define
\eql{
\mathsf D_{\lambda,\sigma}(t;x,y)
:=
\begin{cases}
\mathsf T_{\lambda,\sigma}(t;x,y),
&0<t\leq\varepsilon_\lambda,\\[0.4em]
\mathsf T_{\lambda,\sigma}(t;x,y)-\abs{\rho_\lambda} m_{\lambda,\sigma}(y),
&t>\varepsilon_\lambda.
\end{cases}
\label{eq:defect-kernel}
}
Thus
\eql{
\mathsf T_{\lambda,\sigma}(t;x,y)
=
\chi_{\Set{t>\varepsilon_\lambda}}
\abs{\rho_\lambda} m_{\lambda,\sigma}(y)
+
\mathsf D_{\lambda,\sigma}(t;x,y).
\label{eq:good-plus-defect}
}
Set
\eql{
d_\lambda(t)
:=
\max_{\sigma=\pm1}
\sup_{x\in A_{\lambda,\sigma}\cup\partial\Omega_{\lambda,-\sigma}}
\int_{y\in\partial\Omega_{\lambda,\sigma}}
\abs{\mathsf D_{\lambda,\sigma}(t;x,y)}\dif{y},
\label{eq:def-d-lambda}
}
and
\eql{
\widetilde d_\lambda(t)
:=
\max_{\sigma=\pm1}
\abs{\partial\Omega_{\lambda,\sigma}}
\sup_{\substack{
x\in A_{\lambda,\sigma}\cup\partial\Omega_{\lambda,-\sigma}\\
y\in\partial\Omega_{\lambda,\sigma}
}}
\abs{\mathsf D_{\lambda,\sigma}(t;x,y)}.
\label{eq:def-dtilde-lambda}
}
Then
\eql{
\eta_\lambda
:=
\int_0^\infty d_\lambda(t)\dif{t}
\longrightarrow0,
\label{eq:q-lambda}
}
\eql{
\eta_\lambda^{(1)}
:=
\int_0^\infty t\,d_\lambda(t)\dif{t}
\longrightarrow0,
\label{eq:q1-lambda}
}
and
\eql{
\widehat\eta_\lambda
:=
\sup_{t>0}t\,\widetilde d_\lambda(t)
\longrightarrow0.
\label{eq:qhat-lambda}
}
\end{lem}

\begin{proof}
We treat short and long residence times separately. First suppose that $0<t\leq\varepsilon_\lambda$. Then $\mathsf D_{\lambda,\sigma}=\mathsf T_{\lambda,\sigma}$ and the kernel is nonnegative. Hence
\begin{align}
&\int_0^{\varepsilon_\lambda}
\int_{y\in\partial\Omega_{\lambda,\sigma}}
\mathsf T_{\lambda,\sigma}(t;x,y)\dif{y}\dif{t}
\nonumber\\
&\qquad\leq
\ee^{\varepsilon_\lambda e_\lambda^D}
\frac{
\displaystyle
\sup_{y\in\partial\Omega_{\lambda,\sigma}}
\vf_{\lambda,0}^{\Omega_{\lambda,-\sigma}}(y)
}{
\vf_{\lambda,0}^{\Omega_{\lambda,\sigma}}(x)
}
\nonumber\\
&\qquad\qquad\times
\EE_x\left[
\exp\left(
-\lambda^2
\int_0^{\tau_{\Omega_{\lambda,\sigma}}}
V(\gamma(s))\dif{s}
\right);
\ \tau_{\Omega_{\lambda,\sigma}}\leq\varepsilon_\lambda
\right].
\label{eq:short-defect-first-exit}
\end{align}
The one-well ground-state estimate gives
\eql{
\sup_{\sigma=\pm1}
\sup_{\substack{
x\in A_{\lambda,\sigma}\cup\partial\Omega_{\lambda,-\sigma}\\
y\in\partial\Omega_{\lambda,\sigma}
}}
\frac{
\vf_{\lambda,0}^{\Omega_{\lambda,-\sigma}}(y)
}{
\vf_{\lambda,0}^{\Omega_{\lambda,\sigma}}(x)
}
=
\ee^{o(\lambda)}.
\label{eq:ground-state-ratio-boundary-spheres}
}
Since restricting the exit time can only decrease the expectation, the total first-exit Agmon estimate recorded in \cref{lem:exterior-ball-path-estimates} gives
\eql{
\EE_x\left[
\exp\left(
-\lambda^2
\int_0^{\tau_{\Omega_{\lambda,\sigma}}}
V(\gamma(s))\dif{s}
\right);
\ \tau_{\Omega_{\lambda,\sigma}}\leq\varepsilon_\lambda
\right]
\leq
\exp\left(
-\lambda
S_{\Omega_{\lambda,\sigma}}
\br{x,\partial\Omega_{\lambda,\sigma}}
+o(\lambda)
\right)
\label{eq:short-defect-agmon-upper}
}
uniformly for $x\in A_{\lambda,\sigma}\cup\partial\Omega_{\lambda,-\sigma}$. Concatenating a path from $\sigma d$ to $x$, a path from $x$ to $\partial\Omega_{\lambda,\sigma}$, and a radial segment from the exit point to $-\sigma d$ gives
\eql{
S_{\Omega_{\lambda,\sigma}}
\br{x,\partial\Omega_{\lambda,\sigma}}
\geq
S(d,-d)-Cr_\lambda^2.
\label{eq:boundary-to-boundary-action-lower}
}
Since $\lambda r_\lambda^2=\lambda^{1/2}=o(\lambda)$ and $\varepsilon_\lambda e_\lambda^D=o(\lambda)$, we obtain
\eql{
\int_0^{\varepsilon_\lambda}d_\lambda(t)\dif{t}
\leq
\exp\br{-\lambda S(d,-d)+o(\lambda)}
=
\ee^{o(\lambda)}\abs{\rho_\lambda}
\longrightarrow0.
\label{eq:short-defect-L1}
}
Consequently,
\eql{
\int_0^{\varepsilon_\lambda}t\,d_\lambda(t)\dif{t}
\leq
\varepsilon_\lambda\ee^{o(\lambda)}\ \abs{rho_\lambda}
\longrightarrow0.
\label{eq:short-defect-first-moment}
}
Applying \cref{eq:pointwise-first-exit-bound} with
$T_\lambda=\varepsilon_\lambda$, and then using
\cref{eq:ground-state-ratio-boundary-spheres,eq:boundary-to-boundary-action-lower}
and $\varepsilon_\lambda e_\lambda^D=o(\lambda)$, gives
\eql{
\sup_{0<t\leq\varepsilon_\lambda}
t\,\widetilde d_\lambda(t)
\leq
\exp\br{-\lambda S(d,-d)+o(\lambda)}
\longrightarrow0.
\label{eq:short-defect-pointwise}
}

We next derive the long-time estimate which determines the choice of
$\varepsilon_\lambda$. Set $\tau_\lambda:=\lambda^{-2}$. For
$t>2\tau_\lambda$,
\eql{
\ee^{-t h_\lambda^{\Omega_{\lambda,\sigma}}}
Q_\lambda^{\Omega_{\lambda,\sigma}}
=
\ee^{-\tau_\lambda h_\lambda^{\Omega_{\lambda,\sigma}}}
\left(
\ee^{-(t-2\tau_\lambda)h_\lambda^{\Omega_{\lambda,\sigma}}}
Q_\lambda^{\Omega_{\lambda,\sigma}}
\right)
\ee^{-\tau_\lambda h_\lambda^{\Omega_{\lambda,\sigma}}}.
\label{eq:long-Q-smoothing-factorization}
}
The spectral theorem gives
\eql{
\norm{
\ee^{-(t-2\tau_\lambda)h_\lambda^{\Omega_{\lambda,\sigma}}}
Q_\lambda^{\Omega_{\lambda,\sigma}}
}_{2\to2}
\leq
\ee^{-(t-2\tau_\lambda)\br{e_\lambda^D+\delta_\lambda^D}}.
\label{eq:long-Q-spectral-gap}
}
We quantify the two fixed-time smoothing factors. Feynman--Kac
domination by the free heat kernel gives
\eql{
\sup_{x\in\Omega_{\lambda,\sigma}}
\norm{
K_{\lambda,\tau_\lambda}^{\Omega_{\lambda,\sigma}}(x,\cdot)
}_2^2
=
\sup_{x\in\Omega_{\lambda,\sigma}}
K_{\lambda,2\tau_\lambda}^{\Omega_{\lambda,\sigma}}(x,x)
\leq
(8\pi\tau_\lambda)^{-n/2}.
\label{eq:long-kernel-row-bound}
}
After squaring and integrating in the interior variable, the boundary
Gaussian estimate
\cref{eq:exterior-ball-poisson-kernel-bound}, with
$s=\tau_\lambda$, similarly gives
\eql{
\sup_{y\in\partial\Omega_{\lambda,\sigma}}
\norm{
-\nu_{\lambda,\sigma}(y)\cdot
\nabla_y
K_{\lambda,\tau_\lambda}^{\Omega_{\lambda,\sigma}}
(\cdot,y)
}_2
\leq
C\tau_\lambda^{-n/4-1/2}.
\label{eq:long-flux-row-bound}
}
The constants are uniform in $\lambda$: on the spatial scale
$\sqrt{\tau_\lambda}$ the boundary sphere has radius
$r_\lambda/\sqrt{\tau_\lambda}=\lambda^{3/4}$.
Combining
\cref{eq:long-Q-smoothing-factorization,eq:long-Q-spectral-gap,eq:long-kernel-row-bound,eq:long-flux-row-bound}
and using Cauchy--Schwarz gives
\eql{
&\abs{
-\nu_{\lambda,\sigma}(y)\cdot
\nabla_y
\left(
\ee^{-t h_\lambda^{\Omega_{\lambda,\sigma}}}
Q_\lambda^{\Omega_{\lambda,\sigma}}
\right)(x,y)
}
\nonumber\\
&\qquad\leq
C\lambda^{n+1}
\ee^{-(t-2\tau_\lambda)
\br{e_\lambda^D+\delta_\lambda^D}}.
\label{eq:long-Q-flux-smoothing-bound}
}
Since
\eq{
e_\lambda^D+\delta_\lambda^D
=
e_{\lambda,1}^{\Omega_{\lambda,\sigma}}
<
\lambda^2c_\infty,
}
the factor produced by the two end intervals satisfies
\eq{
\ee^{2\tau_\lambda
\br{e_\lambda^D+\delta_\lambda^D}}
\leq
\ee^{2c_\infty}.
}
Thus \cref{eq:ground-state-ratio-boundary-spheres} and
\cref{eq:long-Q-flux-smoothing-bound} show that the normalized
$Q$-kernel
\eq{
\ee^{t e_\lambda^D}
\frac{
\vf_{\lambda,0}^{\Omega_{\lambda,-\sigma}}(y)
}{
\vf_{\lambda,0}^{\Omega_{\lambda,\sigma}}(x)
}
\left[
-\nu_{\lambda,\sigma}(y)\cdot
\nabla_y
\left(
\ee^{-t h_\lambda^{\Omega_{\lambda,\sigma}}}
Q_\lambda^{\Omega_{\lambda,\sigma}}
\right)(x,y)
\right]
}
has both the $L^1$ boundary norm in \cref{eq:def-d-lambda} and the
pointwise norm in \cref{eq:def-dtilde-lambda} bounded by
$A_\lambda\ee^{-\delta_\lambda^Dt}$ for $t>2\tau_\lambda$, where,
after enlarging it to be at least one,
\eq{
A_\lambda=\ee^{o(\lambda)}.
}
Here we used that $\abs{\partial\Omega_{\lambda,\sigma}}$ and the
powers of $\lambda$ in
\cref{eq:long-Q-flux-smoothing-bound} are polynomial.

For completeness, the same calculation without the boundary normal
derivative gives
\eql{
\abs{
\left(
\ee^{-t h_\lambda^{\Omega_{\lambda,\sigma}}}
Q_\lambda^{\Omega_{\lambda,\sigma}}
\right)(y,z)
}
\leq
C\lambda^n
\ee^{-(t-2\tau_\lambda)
\br{e_\lambda^D+\delta_\lambda^D}}
\label{eq:long-Q-kernel-smoothing-bound}
}
for $y,z\in\Omega_{\lambda,\sigma}$. By
\cref{eq:exterior-ball-ground-state-profile},
\eq{
\sup_{\substack{
y\in\partial\Omega_{\lambda,-\sigma}\\
z\in A_{\lambda,\sigma}
}}
\frac{
\vf_{\lambda,0}^{\Omega_{\lambda,\sigma}}(z)
}{
\vf_{\lambda,0}^{\Omega_{\lambda,\sigma}}(y)
}
=
\ee^{o(\lambda)}.
}
Indeed, both sets remain a fixed positive distance from
$\partial\Omega_{\lambda,\sigma}$, and their constrained actions from
$\sigma d$ are $O(r_\lambda^2)$.
We choose the same $A_\lambda$ large enough that the corresponding
normalized $Q$-kernel in
\cref{eq:long-Q-kernel-smoothing-bound} is also bounded by
$A_\lambda\ee^{-\delta_\lambda^Dt}$. This is the estimate used below
for the final residence interval.

The short-time estimates above hold for every choice of $\varepsilon_\lambda\downarrow0$. Since $\log A_\lambda=o(\lambda)$, we may therefore choose $L_\lambda\to\infty$ so that
\eq{
\log A_\lambda=o(L_\lambda),
\qquad
L_\lambda=o(\lambda),
}
and set
\eq{
\varepsilon_\lambda:=\frac{L_\lambda}{\lambda}.
}
Then $\lambda\varepsilon_\lambda\to\infty$, so $\varepsilon_\lambda>2\tau_\lambda$ for all sufficiently large $\lambda$. For this choice, whenever $t>\varepsilon_\lambda$ the defect is precisely the normalized $Q$-kernel above, and hence
\eql{
d_\lambda(t)+\widetilde d_\lambda(t)
\leq
A_\lambda\ee^{-\delta_\lambda^D t}
\qquad
(t>\varepsilon_\lambda).
\label{eq:long-defect-exponential}
}
Since $\delta_\lambda^D\geq c_{\rm gap}^D\lambda$, this choice gives
\eql{
A_\lambda\ee^{-\delta_\lambda^D\varepsilon_\lambda}
\longrightarrow0.
\label{eq:choice-epsilon-kills-Q}
}
It follows that
\eq{
\int_{\varepsilon_\lambda}^\infty d_\lambda(t)\dif{t}
\leq
\frac{A_\lambda\ee^{-\delta_\lambda^D\varepsilon_\lambda}}{\delta_\lambda^D}
\longrightarrow0,
}
and
\eq{
\int_{\varepsilon_\lambda}^\infty t\,d_\lambda(t)\dif{t}
\leq
A_\lambda\ee^{-\delta_\lambda^D\varepsilon_\lambda}
\left(
\frac{\varepsilon_\lambda}{\delta_\lambda^D}
+
\frac{1}{(\delta_\lambda^D)^2}
\right)
\longrightarrow0.
}
Moreover, since $\delta_\lambda^D\varepsilon_\lambda\to\infty$, the function $t\mapsto t\ee^{-\delta_\lambda^Dt}$ is decreasing for $t\geq\varepsilon_\lambda$ for all sufficiently large $\lambda$, and hence
\eql{
\sup_{t>\varepsilon_\lambda}t\,\widetilde d_\lambda(t)
\leq
A_\lambda\varepsilon_\lambda
\ee^{-\delta_\lambda^D\varepsilon_\lambda}
\longrightarrow0.
\label{eq:long-defect-pointwise}
}
Together with \cref{eq:short-defect-L1,eq:short-defect-first-moment,eq:short-defect-pointwise}, this proves the lemma.
\end{proof}

\subsection{The fixed number of passages}

We first prove the asymptotic for each fixed $k$.

\begin{lem}[Fixed-$k$ factorization]
\label{lem:fixed-k-factorization}
For every fixed $N>0$ and $k\in\NN_{\geq0}$,
\eq{
\frac{
Z_{\lambda,k}\br{\frac{N}{\abs{\rho_\lambda}}}
}{
Z_{\lambda,0}\br{\frac{N}{\abs{\rho_\lambda}}}
}
\longrightarrow
\frac{N^k}{k!}.
}
\end{lem}

\begin{proof}
The result is immediate for $k=0$, so suppose $k\geq1$. Repeated use of the strong Markov property gives the exact $k$-passage analogue of \cref{eq:one passage flux convolution}. Set
\eq{
\sigma_j:=(-1)^j,
\qquad
t_0:=0,
\qquad
x_0:=x.
}
Then
\begin{align}
Z_{\lambda,k}(\beta)
&=
2\int_{x\in A_{\lambda,+}}
\chi_{\lambda,+}(x)^2
\int_{0<t_1<\cdots<t_k<\beta}
\nonumber\\
&\quad\times
\int_{x_1\in\partial\Omega_{\lambda,\sigma_0}}
\cdots
\int_{x_k\in\partial\Omega_{\lambda,\sigma_{k-1}}}
\prod_{j=1}^k
J_{\lambda,\sigma_{j-1}}
\br{t_j-t_{j-1};x_{j-1},x_j}
\nonumber\\
&\quad\times
K_{\lambda,\beta-t_k}^{\Omega_{\lambda,\sigma_k}}
\br{x_k,\sigma_kx}
\dif{x_1}\cdots\dif{x_k}
\dif{t_1}\cdots\dif{t_k}\dif{x}.
\label{eq:general-k-passage-Markov-formula}
\end{align}
Using \cref{eq:normalized-passage-kernel}, multiplying and dividing by the corresponding one-well ground states in each passage factor, and using reflection symmetry at the endpoint, we may rewrite \cref{eq:general-k-passage-Markov-formula} exactly as
\begin{align}
Z_{\lambda,k}(\beta)
&=
2\ee^{-\beta e_\lambda^D}
\int_{x\in A_{\lambda,+}}
\chi_{\lambda,+}(x)^2
\int_{0<t_1<\cdots<t_k<\beta}
\nonumber\\
&\quad\times
\int_{x_1\in\partial\Omega_{\lambda,\sigma_0}}
\cdots
\int_{x_k\in\partial\Omega_{\lambda,\sigma_{k-1}}}
\prod_{j=1}^k
\mathsf T_{\lambda,\sigma_{j-1}}
\br{t_j-t_{j-1};x_{j-1},x_j}
\nonumber\\
&\quad\times
\left[
\ee^{(\beta-t_k)e_\lambda^D}
\frac{
\vf_{\lambda,0}^{\Omega_{\lambda,\sigma_k}}(\sigma_kx)
}{
\vf_{\lambda,0}^{\Omega_{\lambda,\sigma_k}}(x_k)
}
K_{\lambda,\beta-t_k}^{\Omega_{\lambda,\sigma_k}}
\br{x_k,\sigma_kx}
\right]
\nonumber\\
&\quad\times
\dif{x_1}\cdots\dif{x_k}
\dif{t_1}\cdots\dif{t_k}\dif{x}.
\label{eq:general-k-normalized-passage-formula}
\end{align}
The expression in square brackets is the ground-state transform of the final one-well heat kernel. It is nonnegative and, after extending the $x$-integration to the whole corresponding one-well domain, integrates to one. Moreover, the same smoothing argument used in \cref{lem:integrated-defect-estimate} gives, with the same $A_\lambda$, uniformly for $s\geq\varepsilon_\lambda$ and $y\in\partial\Omega_{\lambda,-\sigma}$,
\eql{
\int_{x\in A_{\lambda,+}}
\chi_{\lambda,+}(x)^2
\ee^{s e_\lambda^D}
\frac{
\vf_{\lambda,0}^{\Omega_{\lambda,\sigma}}(\sigma x)
}{
\vf_{\lambda,0}^{\Omega_{\lambda,\sigma}}(y)
}
K_{\lambda,s}^{\Omega_{\lambda,\sigma}}(y,\sigma x)
\dif{x}
=
a_\lambda
+\calO\br{A_\lambda\ee^{-\delta_\lambda^D s}}.
\label{eq:final-ground-state-mixing}
}
Indeed, the ground-state part is $a_\lambda$. For the orthogonal part,
set $z=\sigma x$ in the normalized estimate following
\cref{eq:long-Q-kernel-smoothing-bound} and integrate in $x$; since
\eq{
\int_{x\in A_{\lambda,+}}
\chi_{\lambda,+}(x)^2\dif{x}
\leq
C,
}
this gives the stated error.

Let
\eq{
\mathcal G_{\lambda,k}(\beta)
:=
\Set{
0<t_1<\cdots<t_k<\beta
\ \big|\
t_1>\varepsilon_\lambda,\
t_j-t_{j-1}>\varepsilon_\lambda,\
\beta-t_k>\varepsilon_\lambda
}.
}
Restrict \cref{eq:general-k-normalized-passage-formula} to this region and insert \cref{eq:good-plus-defect} in each of the $k$ passage factors. The term in which every passage factor is good equals
\eql{
2\ee^{-\beta e_\lambda^D}
\abs{\rho_\lambda}^k
\br{a_\lambda+o(1)}
\frac{\br{\beta-(k+1)\varepsilon_\lambda}_+^k}{k!},
\label{eq:all-P-separated-contribution}
}
where the $o(1)$ is uniform on the restricted time simplex by \cref{eq:final-ground-state-mixing,eq:choice-epsilon-kills-Q}. Every boundary integration of a good passage factor gives one because $m_{\lambda,\sigma}$ is a probability density.

It remains to control the terms containing defects. Suppose first that exactly $r\in\Set{1,\ldots,k-1}$ of the $k$ passage factors are defects. Cut the spatial chain at one of the good rank-one passage factors. Using the $L^1$ boundary norm in \cref{eq:def-d-lambda}, the fact that $m_{\lambda,\sigma}$ is a probability density, and the Markov property of the final ground-state-transformed kernel, the absolute value of each such word is bounded by
\eq{
2\ee^{-\beta e_\lambda^D}
\abs{\rho_\lambda}^{k-r}
\frac{\beta^{k-r}}{(k-r)!}
\eta_\lambda^r.
}
Here the factor $\beta^{k-r}/(k-r)!$ comes from the $k-r$ good passage durations together with the final residence time, while each defect duration is integrated using \cref{eq:q-lambda}. Thus, at $\beta=N/\abs{\rho_\lambda}$, the terms with at least one good passage factor are bounded in total by
\eq{
2\ee^{-\beta e_\lambda^D}
\sum_{r=1}^{k-1}
\binom{k}{r}
\frac{N^{k-r}}{(k-r)!}
\eta_\lambda^r
=
o\br{\ee^{-\beta e_\lambda^D}}
}
for fixed $k$.

If every passage factor is a defect, use the pointwise bound for the first defect and the $L^1$ boundary bounds for the remaining defects. Since all passage durations on $\mathcal G_{\lambda,k}(\beta)$ are at least $\varepsilon_\lambda$, \cref{eq:long-defect-exponential} gives
\eq{
\int_{\varepsilon_\lambda}^{\infty}
\widetilde d_\lambda(t)\dif{t}
\leq
\frac{A_\lambda\ee^{-\delta_\lambda^D\varepsilon_\lambda}}{\delta_\lambda^D}
=o(1).
}
The absolute value of the all-defect word is therefore bounded by
\eq{
2\ee^{-\beta e_\lambda^D}
\left(
\int_{\varepsilon_\lambda}^{\infty}
\widetilde d_\lambda(t)\dif{t}
\right)
\eta_\lambda^{k-1}
=
o\br{\ee^{-\beta e_\lambda^D}}.
}
Consequently the contribution of $\mathcal G_{\lambda,k}(\beta)$ satisfies
\begin{align}
\frac{Z_{\lambda,k}^{\mathcal G}(\beta)}{Z_{\lambda,0}(\beta)}
&=
\frac{\br{\beta-(k+1)\varepsilon_\lambda}_+^k\abs{\rho_\lambda}^k}{k!}
+o(1).
\label{eq:separated-lower-asymptotic}
\end{align}
Here we used \cref{eq:zero passage flux asymptotic} and the one-well localization estimate
\eql{
a_\lambda\longrightarrow1.
\label{eq:a-lambda-to-one}
}
At $\beta=N/\abs{\rho_\lambda}$ we have $\varepsilon_\lambda\abs{\rho_\lambda}\to0$, and therefore
\eql{
\liminf_{\lambda\to\infty}
\frac{Z_{\lambda,k}\br{\frac{N}{\abs{\rho_\lambda}}}}{Z_{\lambda,0}\br{\frac{N}{\abs{\rho_\lambda}}}}
\geq
\frac{N^k}{k!}.
\label{eq:fixed-k-liminf}
}

It remains to prove the reverse inequality. In \cref{eq:general-k-passage-Markov-formula}, use positivity to replace $\chi_{\lambda,+}^2$ by $1$ and extend the $x$-integration to the corresponding one-well domain. Reflect the final heat kernel when $k$ is odd. Symmetry of the Dirichlet heat kernel and the semigroup property give
\begin{align}
&\int_{x\in\Omega_\lambda}
K_{\lambda,\beta-t_k}^{\Omega_{\lambda,\sigma_k}}
\br{x_k,\sigma_kx}
J_{\lambda,+1}(t_1;x,x_1)
\dif{x}
\nonumber\\
&\qquad=
J_{\lambda,+1}
\br{\beta-t_k+t_1;\sigma_kx_k,x_1}.
\label{eq:cyclic-gluing-identity}
\end{align}
Set
\eq{
u_j:=t_{j+1}-t_j
\quad(j=1,\ldots,k-1),
\qquad
u_k:=\beta-t_k+t_1.
}
Then $u_j>0$ and $u_1+\cdots+u_k=\beta$. For fixed $u_1,\ldots,u_k$, the variable $t_1$ ranges over an interval of length $u_k$. Thus
\eql{
Z_{\lambda,k}(\beta)
\leq
2\ee^{-\beta e_\lambda^D}
\int_{\substack{
u_1,\ldots,u_k>0\\
u_1+\cdots+u_k=\beta
}}
u_k\,
\mathfrak C_{\lambda,k}(u_1,\ldots,u_k)
\dif{u_1}\cdots\dif{u_{k-1}},
\label{eq:cyclic-upper-bound}
}
where $\mathfrak C_{\lambda,k}$ is the cyclic boundary integral of the $k$ alternating normalized kernels \cref{eq:normalized-passage-kernel}; in the last boundary variable we apply reflection when $k$ is odd. The ground-state factors telescope around the cycle.

Insert \cref{eq:good-plus-defect} in every cyclic factor. For fixed $k$, any term containing at least one defect tends to zero after the time integrations by \cref{eq:q-lambda,eq:q1-lambda,eq:qhat-lambda}. The unique term containing only good factors is bounded above by
\eq{
\abs{\rho_\lambda}^k\frac{\beta^k}{k!}.
}
Thus
\eq{
\limsup_{\lambda\to\infty}
\frac{Z_{\lambda,k}\br{\frac{N}{\abs{\rho_\lambda}}}}{Z_{\lambda,0}\br{\frac{N}{\abs{\rho_\lambda}}}}
\leq
\frac{N^k}{k!},
}
where we again used \cref{eq:a-lambda-to-one}. Together with \cref{eq:fixed-k-liminf}, this proves the lemma.
\end{proof}

\subsection{A summable estimate uniform in the number of passages}

We keep the complete defect expansion in \cref{eq:cyclic-upper-bound} and estimate it uniformly in $k$.

\begin{lem}[Uniform summable bound]
\label{lem:uniform-summable-passage-bound}
For every fixed $N>0$, there exists $C_N<\infty$ such that, for all sufficiently large $\lambda$,
\eq{
\frac{Z_{\lambda,k}\br{\frac{N}{\abs{\rho_\lambda}}}}{Z_{\lambda,0}\br{\frac{N}{\abs{\rho_\lambda}}}}
\leq
C_N2^{-k}
\qquad
(k\in\NN_{\geq0}).
}
\end{lem}

\begin{proof}
The assertion is trivial for $k=0$, so suppose $k\geq1$. Expand each factor in \cref{eq:cyclic-upper-bound} according to \cref{eq:good-plus-defect}.

Suppose first that the distinguished factor corresponding to $u_k$ is good and that the word contains exactly $m$ good factors. There are
\eq{
\binom{k-1}{m-1}
}
such words. Cutting the cyclic product at the distinguished rank-one kernel, and using
\eq{
\int_{y\in\partial\Omega_{\lambda,\sigma}}m_{\lambda,\sigma}(y)\dif{y}=1,
}
shows that its spatial integral is bounded by
\eq{
\abs{\rho_\lambda}^m
\prod_{\substack{j:\ \text{$j$th factor is a defect}}}
d_\lambda(u_j).
}
After integrating the defect durations and then the $m$ good durations, the time integral is bounded by
\eql{
\abs{\rho_\lambda}^m
\frac{\beta^m}{m!}
\eta_\lambda^{k-m}.
\label{eq:word-bound-distinguished-good}
}
The factor $\beta^m/m!$ is the convolution of $m-1$ copies of the constant function $1$ and one copy of the function $t$ coming from the distinguished factor $u_k$.

Next suppose that $u_k$ is a defect, while exactly $m\geq1$ of the remaining $k-1$ factors are good. There are
\eq{
\binom{k-1}{m}
}
such words. Cutting the cycle at any good factor and absorbing the factor $u_k$ with $\eta_\lambda^{(1)}$ gives
\eql{
\abs{\rho_\lambda}^m
\frac{\beta^{m-1}}{(m-1)!}
\eta_\lambda^{(1)}
\eta_\lambda^{k-m-1}.
\label{eq:word-bound-distinguished-defect}
}

Finally suppose that every factor is a defect. Use the pointwise bound for the distinguished last kernel and the $L^1$ boundary bounds for the remaining $k-1$ kernels. The cyclic boundary integral is bounded by
\eq{
\widetilde d_\lambda(u_k)
\prod_{j=1}^{k-1}d_\lambda(u_j).
}
Since $u_k\widetilde d_\lambda(u_k)\leq\widehat\eta_\lambda$, integration over the remaining $k-1$ variables gives
\eql{
\widehat\eta_\lambda \eta_\lambda^{k-1}.
\label{eq:all-defect-bound}
}
This is the point at which the pointwise estimate \cref{eq:qhat-lambda} is needed; an integrated first-moment estimate alone would not control the convolution at the prescribed total time $\beta$.

Since
\eq{
Z_{\lambda,0}(\beta)
\geq
2a_\lambda\ee^{-\beta e_\lambda^D}
}
and $a_\lambda\to1$, we may assume $a_\lambda\geq\frac12$ for all sufficiently large $\lambda$. Setting $\beta=N/\abs{\rho_\lambda}$ and summing \cref{eq:word-bound-distinguished-good,eq:word-bound-distinguished-defect,eq:all-defect-bound} over the possible positions of the good factors gives
\begin{align}
\frac{Z_{\lambda,k}\br{\frac{N}{\abs{\rho_\lambda}}}}{Z_{\lambda,0}\br{\frac{N}{\abs{\rho_\lambda}}}}
&\leq
2\sum_{m=1}^k
\binom{k-1}{m-1}
\frac{N^m}{m!}
\eta_\lambda^{k-m}
\nonumber\\
&\quad+
2\abs{\rho_\lambda} \eta_\lambda^{(1)}
\sum_{m=1}^{k-1}
\binom{k-1}{m}
\frac{N^{m-1}}{(m-1)!}
\eta_\lambda^{k-m-1}
\nonumber\\
&\quad+
2\widehat\eta_\lambda \eta_\lambda^{k-1}.
\label{eq:uniform-k-combinatorial-bound}
\end{align}
The right-hand side has a uniformly summable generating function. Indeed,
\begin{align}
&\sum_{k=1}^\infty z^k
\sum_{m=1}^k
\binom{k-1}{m-1}
\frac{N^m}{m!}
\eta_\lambda^{k-m}
\nonumber\\
&\qquad=
\exp\left(\frac{Nz}{1-\eta_\lambda z}\right)-1,
\label{eq:generating-function-first-sum}
\end{align}
\begin{align}
&\sum_{k=2}^\infty z^k
\sum_{m=1}^{k-1}
\binom{k-1}{m}
\frac{N^{m-1}}{(m-1)!}
\eta_\lambda^{k-m-1}
\nonumber\\
&\qquad=
\frac{z^2}{(1-\eta_\lambda z)^2}
\exp\left(\frac{Nz}{1-\eta_\lambda z}\right),
\label{eq:generating-function-second-sum}
\end{align}
and
\eql{
\sum_{k=1}^\infty
z^k\widehat\eta_\lambda \eta_\lambda^{k-1}
=
\widehat\eta_\lambda\frac{z}{1-\eta_\lambda z}.
\label{eq:generating-function-all-defect}
}
By \cref{eq:q-lambda,eq:q1-lambda,eq:qhat-lambda}, for all sufficiently large $\lambda$,
\eq{
\eta_\lambda\leq\frac14,
\qquad
\abs{\rho_\lambda }\eta_\lambda^{(1)}\leq1,
\qquad
\widehat\eta_\lambda\leq1.
}
Evaluating \cref{eq:generating-function-first-sum,eq:generating-function-second-sum,eq:generating-function-all-defect} at $z=2$ gives
\eq{
\sum_{k=0}^\infty
2^k
\frac{Z_{\lambda,k}\br{\frac{N}{\abs{\rho_\lambda}}}}{Z_{\lambda,0}\br{\frac{N}{\abs{\rho_\lambda}}}}
\leq
C_N
}
with $C_N$ independent of $\lambda$. Since all coefficients are nonnegative,
\eq{
\frac{Z_{\lambda,k}\br{\frac{N}{\abs{\rho_\lambda}}}}{Z_{\lambda,0}\br{\frac{N}{\abs{\rho_\lambda}}}}
\leq
C_N2^{-k}.
}
This proves the lemma.
\end{proof}

\begin{proof}[Proof of \cref{thm:Poisson distribution}]
The existence and positivity of the limit defining $-\rho_\lambda$ have already been proved. Moreover, \cref{eq:hopping-exponent} gives
\eq{
\abs{\rho_\lambda}
=
\exp\br{-\lambda S(d,-d)+o(\lambda)},
}
and hence
\eq{
\log(\lambda)\abs{\rho_\lambda}\longrightarrow0.
}
Fix $N>0$. By \cref{lem:fixed-k-factorization}, for every $\ell\in\NN_{\geq0}$,
\eq{
\frac{Z_{\lambda,\ell}\br{\frac{N}{\abs{\rho_\lambda}}}}{Z_{\lambda,0}\br{\frac{N}{\abs{\rho_\lambda}}}}
\longrightarrow
\frac{N^\ell}{\ell!}.
}
By \cref{lem:uniform-summable-passage-bound}, for all sufficiently large $\lambda$,
\eq{
\frac{Z_{\lambda,\ell}\br{\frac{N}{\abs{\rho_\lambda}}}}{Z_{\lambda,0}\br{\frac{N}{\abs{\rho_\lambda}}}}
\leq
C_N2^{-\ell},
}
and $\sum_{\ell=0}^\infty C_N2^{-\ell}<\infty$. Dominated convergence therefore gives
\eq{
\sum_{\ell=0}^\infty
\frac{Z_{\lambda,\ell}\br{\frac{N}{\abs{\rho_\lambda}}}}{Z_{\lambda,0}\br{\frac{N}{\abs{\rho_\lambda}}}}
\longrightarrow
\sum_{\ell=0}^\infty\frac{N^\ell}{\ell!}
=
\ee^N.
}
Consequently, for every fixed $k\in\NN_{\geq0}$,
\begin{align}
\mathbb P_{\lambda,\frac{N}{\abs{\rho_\lambda}}}^{\rm instanton}\left[\Set{k}\right]
&=
\frac{
\displaystyle
\frac{Z_{\lambda,k}\br{\frac{N}{\abs{\rho_\lambda}}}}{Z_{\lambda,0}\br{\frac{N}{\abs{\rho_\lambda}}}}
}{
\displaystyle
\sum_{\ell=0}^\infty
\frac{Z_{\lambda,\ell}\br{\frac{N}{\abs{\rho_\lambda}}}}{Z_{\lambda,0}\br{\frac{N}{\abs{\rho_\lambda}}}}
}
\nonumber\\
&\longrightarrow
\frac{N^k/k!}{\ee^N}
=
\ee^{-N}\frac{N^k}{k!}.
\end{align}
This proves \cref{thm:Poisson distribution}.
\end{proof}
\appendix

\section{Logarithmic estimates for the exterior-ball domains}
\label{sec:exterior-ball-path-estimates}

We state the estimates only for
\eq{
\Omega_\lambda
=
\RR^n\setminus A_{\lambda,-},
\qquad
-\Omega_\lambda
=
\RR^n\setminus A_{\lambda,+},
\qquad
r_\lambda=\lambda^{-1/4}.
}
When $n=1$, each domain is understood to mean the component containing
its corresponding well. Let $\PP$ be Wiener measure for Brownian motion
with generator $\Delta$, let $\EE$ be its expectation, and write
\eq{
\EE_z[F(\gamma)]
:=
\EE[F(z+\gamma)].
}
For $D_\lambda\in\Set{\Omega_\lambda,-\Omega_\lambda}$, set
\eq{
\tau_{D_\lambda}(\gamma)
:=
\inf\Set{
t>0
\mid
\gamma(t)\notin D_\lambda
}.
}

\begin{lem}[Exterior-ball logarithmic estimates]
\label[lem]{lem:exterior-ball-path-estimates}
Let
\eq{
(D_\lambda,q)
\in
\Set{
(\Omega_\lambda,d),
(-\Omega_\lambda,-d)
}.
}
The following estimates hold.

\begin{enumerate}[label=\textnormal{(\roman*)}]

\item
For every fixed compact set $K\subseteq\RR^n$ and every
$0<T_0<L<\infty$,
\eql{
-\frac1\lambda
\log
K_{\lambda,T/\lambda}^{D_\lambda}(x,y)
=
S_{T,D_\lambda}(x,y)+o(1)
\label{eq:exterior-ball-kernel}
}
uniformly for $T_0\leq T\leq L$ and
\eq{
x,y\in K\cap D_\lambda,
\qquad
\dist\br{
\Set{x,y},
\partial D_\lambda
}
\geq r_\lambda.
}

\item
Suppose that, for some fixed $\eta\in(0,\frac12)$,
\eq{
c\lambda
\leq
e_{\lambda,0}^{D_\lambda}
\leq
C\lambda,
\qquad
\sup_{x\in B_{\lambda^{-1/2+\eta}}(q)}
\frac1\lambda
\log
\vf_{\lambda,0}^{D_\lambda}(x)
\longrightarrow0.
}
Then, for every fixed compact set $K\subseteq\RR^n$,
\eql{
\sup_{\substack{
x\in K\cap D_\lambda\\
\dist(x,\partial D_\lambda)\geq r_\lambda
}}
\left|
-\frac1\lambda
\log
\vf_{\lambda,0}^{D_\lambda}(x)
-
S_{D_\lambda}(q,x)
\right|
\longrightarrow0.
\label{eq:exterior-ball-ground-state-profile}
}
Moreover,
\eq{
\limsup_{\lambda\to\infty}
\sup_{x\in K\cap D_\lambda}
\left\{
\frac1\lambda
\log
\vf_{\lambda,0}^{D_\lambda}(x)
+
S_{D_\lambda}(q,x)
\right\}
\leq0,
}
the estimate in
\cref{eq:exterior-ball-ground-state-profile} is uniform on compact
families of hypersurfaces whose distance from
$\partial D_\lambda$ is at least $r_\lambda$, and
\eq{
\norm{
\vf_{\lambda,0}^{D_\lambda}
}_\infty
\leq
C\lambda^{n/4}.
}

\item
Assume the hypotheses in \textnormal{(ii)} for both
$\Omega_\lambda$ and $-\Omega_\lambda$. For every fixed $L>0$,
\eql{
&-\frac1\lambda
\log
\int_{z\in\Omega_\lambda}
\vf_{\lambda,0}^{\Omega_\lambda}(z)
\EE_z\left[
\exp\left(
-\lambda^2
\int_0^{\tau_{\Omega_\lambda}(\gamma)}
V\br{\gamma(t)}
\dif{t}
\right)
\right.
\nonumber\\
&\hspace{14em}\left.
\times
\vf_{\lambda,0}^{-\Omega_\lambda}
\br{
\gamma\br{
\tau_{\Omega_\lambda}(\gamma)
}
}
\chi_{\Set{
\tau_{\Omega_\lambda}(\gamma)
\leq
L/\lambda
}}
\right]
\dif{z}
\nonumber\\
&\qquad=
\inf_{\substack{
z\in\Omega_\lambda,\ 
y\in\partial\Omega_\lambda\\
0<T\leq L
}}
\left\{
S_{\Omega_\lambda}(d,z)
+
S_{T,\Omega_\lambda}(z,y)
+
S_{-\Omega_\lambda}(y,-d)
\right\}
+o(1).
\label{eq:exterior-ball-integrated-exit}
}
The boxwise upper bounds used to prove
\cref{eq:exterior-ball-integrated-exit} hold uniformly after restricting
the starting points, exit points, or exit times to any of the
finite-dimensional blocks described below, and incur only
$\exp(o(\lambda))$ losses. A matching boxwise lower bound holds whenever
the starting block is at distance at least $r_\lambda$ from
$\partial\Omega_\lambda$ and the restricted variational problem admits
an admissible comparison path. In addition, uniformly for the compact families of starting points used below,
\eq{
\EE_x\left[
\exp\left(
-\lambda^2\int_0^{\tau_{D_\lambda}}V(\gamma(t))\dif{t}
\right);
\ \tau_{D_\lambda}<\infty
\right]
\leq
\exp\left(
-\lambda S_{D_\lambda}\br{x,\partial D_\lambda}+o(\lambda)
\right).
}
For every fixed compact set $K\subseteq\RR^n$, every $\delta>0$,
and every $T_\lambda\downarrow0$, the corresponding pointwise
first-exit density satisfies the same exponential upper bound, up to
an $\exp(o(\lambda))$ factor after multiplication by $t$ and by the
surface area of $\partial D_\lambda$, uniformly for
\eq{
x\in K\cap D_\lambda,
\qquad
\dist\br{x,\partial D_\lambda}\geq\delta,
\qquad
y\in\partial D_\lambda,
\qquad
0<t\leq T_\lambda.
}

\end{enumerate}
\end{lem}

\begin{proof}
We first prove the finite-time path estimate. In the rescaled time
variable $s=\lambda t$, the process
\eq{
X_\lambda(s)
:=
\gamma(s/\lambda),
\qquad
0\leq s\leq L,
}
has increments with density
\eq{
\left(
\frac{\lambda}{4\pi(s-r)}
\right)^{\frac n2}
\exp\left(
-\frac{\lambda\norm{x-y}^2}
{4(s-r)}
\right).
}
Under this change of variables, the Feynman--Kac weight becomes
\eql{
&\exp\left(
-\lambda^2
\int_0^{T/\lambda}
V\br{\gamma(t)}
\dif{t}
\right)
\nonumber\\
&\qquad=
\exp\left(
-\lambda
\int_0^T
V\br{X_\lambda(s)}
\dif{s}
\right).
\label{eq:exterior-ball-rescaled-feynman-kac}
}
Thus the rate functional is
\eq{
\gamma
\longmapsto
\int_0^T
\left(
\frac14
\norm{\dot\gamma(s)}^2
+
V\br{\gamma(s)}
\right)
\dif{s}.
}

Choose
\eq{
a_\lambda
:=
\lambda^{-3/4},
\qquad
\widetilde r_\lambda
:=
\lambda^{-1/3},
\qquad
b_\lambda
:=
a_\lambda^2
=
\lambda^{-3/2}.
}
These scales satisfy
\eql{
\frac{\log\lambda}{\lambda}
\ll
a_\lambda
\ll
r_\lambda^2,
\qquad
\sqrt{a_\lambda}
\ll
\widetilde r_\lambda
\ll
r_\lambda,
\qquad
b_\lambda
\ll
\widetilde r_\lambda.
\label{eq:exterior-ball-mesh-scales}
}
Divide $[0,T]$ into intervals whose lengths lie between
$a_\lambda/2$ and $a_\lambda$. If
\eq{
0=s_0<s_1<\cdots<s_m=T
}
are the resulting mesh points, then, uniformly for
$T\in[T_0,L]$,
\eq{
m=O(a_\lambda^{-1}),
\qquad
m\log\lambda=o(\lambda).
}
The joint density of the mesh values
$x_j=X_\lambda(s_j)$ is
\eql{
&\prod_{j=1}^m
\left(
\frac{\lambda}
{4\pi(s_j-s_{j-1})}
\right)^{\frac n2}
\nonumber\\
&\qquad\times
\exp\left[
-\lambda
\sum_{j=1}^m
\frac{
\norm{x_j-x_{j-1}}^2
}{
4(s_j-s_{j-1})
}
\right]
\dif{x_1}
\cdots
\dif{x_m}.
\label{eq:exterior-ball-mesh-density}
}
The exponent in
\cref{eq:exterior-ball-mesh-density} is precisely the kinetic action
of the polygonal interpolation, whereas every normalization factor
produced by the discretization is $\exp(o(\lambda))$.

Let $X_\lambda^{(m)}$ denote the polygonal interpolation of the mesh
values. Conditionally on those values,
$X_\lambda-X_\lambda^{(m)}$ is a collection of independent Brownian
bridges. The reflection estimate and a union bound give
\eql{
\PP\left[
\norm{
X_\lambda-X_\lambda^{(m)}
}_\infty
>
\widetilde r_\lambda
\right]
\leq
Cm
\exp\left(
-\frac{
c\lambda\widetilde r_\lambda^2
}{
a_\lambda
}
\right).
\label{eq:exterior-ball-bridge-bound}
}
By \cref{eq:exterior-ball-mesh-scales},
\eq{
\frac{
\lambda\widetilde r_\lambda^2
}{
a_\lambda
}
=
\lambda^{13/12}.
}
Consequently, the right-hand side of
\cref{eq:exterior-ball-bridge-bound} is superexponentially small on
the scale $\lambda$.

We next control the error caused by the shrinking boundary. For
$n\geq2$, every point in a neighborhood of the excluded ball can be
written as
\eq{
x=-q+s\omega,
\qquad
s>0,
\qquad
\omega\in\bS^{n-1}.
}
Choose $\vartheta\in C^\infty(\RR;[0,1])$ such that
\eq{
\vartheta(u)=1
\quad
\text{for }
u\leq\frac14,
\qquad
\vartheta(u)=0
\quad
\text{for }
u\geq1.
}
For $0<\delta=o(r_\lambda)$ and
$s\geq r_\lambda-\delta$, define
\eq{
f_{\lambda,\delta}(s)
:=
s
+
3\delta
\vartheta\left(
\frac{s-r_\lambda}{r_\lambda}
\right).
}
For all sufficiently large $\lambda$, this function is increasing and
satisfies
\eq{
f_{\lambda,\delta}(s)
\geq
r_\lambda+2\delta
\qquad
\left(
s\geq r_\lambda-\delta
\right).
}
Indeed, when
$r_\lambda-\delta\leq s\leq r_\lambda+2\delta$, the cutoff equals one
for all sufficiently large $\lambda$, while for
$s\geq r_\lambda+2\delta$ the inequality follows from
$f_{\lambda,\delta}(s)\geq s$. Moreover,
\eq{
f_{\lambda,\delta}(s)
=
s
\quad
\text{for }
s\geq2r_\lambda,
}
and
\eq{
\sup_{s\geq r_\lambda-\delta}
\abs{
f_{\lambda,\delta}(s)-s
}
\leq
3\delta,
\qquad
\sup_{s\geq r_\lambda-\delta}
\abs{
f_{\lambda,\delta}'(s)-1
}
\leq
C\frac{\delta}{r_\lambda}.
}
Define
\eq{
\mathcal R_{\lambda,\delta}
\br{-q+s\omega}
:=
-q+
f_{\lambda,\delta}(s)\omega
}
on the $\delta$-neighborhood of $\overline{D_\lambda}$. This map
takes that neighborhood into
\eq{
\Set{
x\in D_\lambda
\mid
\dist(x,\partial D_\lambda)
\geq2\delta
}
}
and is the identity whenever
$\dist(x,\partial D_\lambda)\geq r_\lambda$. Since
$s\geq r_\lambda-\delta$ on its domain,
\eq{
\sup
\norm{
D\mathcal R_{\lambda,\delta}-1
}
\leq
C\frac{\delta}{r_\lambda}.
}
The corresponding construction in dimension one is immediate.

If $\gamma$ remains in a fixed compact set and has action bounded by
$M$, then
\eql{
&
\left|
\int_0^T
\left(
\frac14
\norm{
\frac{\dif}{\dif{s}}
\mathcal R_{\lambda,\delta}
\br{\gamma(s)}
}^2
+
V\br{
\mathcal R_{\lambda,\delta}
\br{\gamma(s)}
}
\right)
\dif{s}
\right.
\nonumber\\
&\hspace{6em}\left.
-
\int_0^T
\left(
\frac14
\norm{\dot\gamma(s)}^2
+
V\br{\gamma(s)}
\right)
\dif{s}
\right|
\nonumber\\
&\qquad\leq
C_M
\left(
\frac{\delta}{r_\lambda}
+
\delta
\right).
\label{eq:exterior-ball-radial-deformation}
}
Here we used the Lipschitz continuity of $V$ on compact sets. In
particular, taking $\delta=\widetilde r_\lambda$ in
\cref{eq:exterior-ball-radial-deformation} changes the action by
$o(1)$.

Paths which leave a sufficiently large fixed ball may be discarded.
Indeed, the finite-dimensional Gaussian density gives an upper bound
of the form
\eq{
C
\exp\left(
-\frac{c\lambda R^2}{L}
\right)
}
for paths which start and end in a fixed compact set but leave
$B_R(0)$. We may therefore perform the remaining argument inside a
fixed compact set.

Cover this compact set by cubes of side $b_\lambda$. The logarithm
of the number of possible strings of cubes is
\eq{
O\br{
m\log(b_\lambda^{-1})
}
=
o(\lambda).
}
On bounded action sets, replacing each mesh value by any point in
the same cube changes the kinetic action by $o(1)$. Indeed, the
Cauchy--Schwarz inequality reduces the error to
\eq{
C
\left(
\sum_{j=1}^m
\frac{
\norm{x_j-x_{j-1}}^2
}{
s_j-s_{j-1}
}
\right)^{1/2}
\left(
\sum_{j=1}^m
\frac{
b_\lambda^2
}{
s_j-s_{j-1}
}
\right)^{1/2}
+
C
\sum_{j=1}^m
\frac{
b_\lambda^2
}{
s_j-s_{j-1}
},
}
and
\eq{
\sum_{j=1}^m
\frac{
b_\lambda^2
}{
s_j-s_{j-1}
}
=
O\left(
\frac{
b_\lambda^2
}{
a_\lambda^2
}
\right)
=
O(a_\lambda^2)
=
o(1).
}
The error in the potential term is also $o(1)$. This follows from
the uniform continuity of $V$ on the compact set and the
Cauchy--Schwarz bound on the oscillation of a bounded-action
polygonal path over one mesh interval.

For the upper bound in
\cref{eq:exterior-ball-kernel}, sum over the cube strings and use
\cref{eq:exterior-ball-bridge-bound}. If the Brownian path stays in
$D_\lambda$, then, outside a superexponentially small event, its
polygonal interpolation stays within
$\widetilde r_\lambda$ of $D_\lambda$. Applying
$\mathcal R_{\lambda,\widetilde r_\lambda}$ and
\cref{eq:exterior-ball-radial-deformation} produces an admissible
path whose action differs by $o(1)$. Since the number of cube strings
is $\exp(o(\lambda))$, this proves the required upper bound.

For the lower bound, choose a polygonal path whose action is within
$o(1)$ of $S_{T,D_\lambda}(x,y)$ and apply the radial deformation
with $3\widetilde r_\lambda$ in place of
$\widetilde r_\lambda$. Because the endpoints are at distance at
least $r_\lambda$ from the boundary and
$\widetilde r_\lambda\ll r_\lambda$, the deformation leaves the
endpoints fixed. Restrict the mesh values to cubes of side
$b_\lambda$ about the deformed path and restrict every Brownian bridge
to remain within $\widetilde r_\lambda$ of its linear interpolation.
All paths in this tube remain in $D_\lambda$. The mesh volumes,
Gaussian normalization factors, and bridge restrictions together
contribute only $\exp(o(\lambda))$. This proves the lower bound and
hence \cref{eq:exterior-ball-kernel}.

We next establish the corresponding first-exit estimate. The scale
$b_\lambda$ will continue to be used for the mesh-node integrations,
but it is unnecessarily small for the exit-time and exit-point
decompositions. Set
\eq{
\ell_\lambda
:=
\widetilde r_\lambda.
}
Cover $(0,L]$ by overlapping intervals of length
$O(\ell_\lambda)$ whose concentric subintervals of half that length
still cover $(0,L]$. Likewise, cover $\partial D_\lambda$ by
overlapping coordinate patches of diameter $O(\ell_\lambda)$ whose
concentric subpatches of half that diameter still cover the boundary.
The total number of resulting space-time blocks is bounded by
\eq{
C
\ell_\lambda^{-1}
\left(
1+
\frac{r_\lambda}{\ell_\lambda}
\right)^{n-1}
=
\exp(o(\lambda)).
}
In dimension one there is only one boundary point for the relevant
component.

Fix one exit-time block and one boundary patch. We first prove the
upper bound. Sort the paths according to the fine mesh interval
$[s_{j-1},s_j]$ containing their first exit. On the complement of
the event in
\cref{eq:exterior-ball-bridge-bound}, the polygonal interpolation is
within $\ell_\lambda$ of the Brownian path. Consequently, if the true
path exits through the prescribed patch during the prescribed time
block, then the polygonal path reaches the
$C\ell_\lambda$-enlargement of that patch during the
$C\ell_\lambda$-enlargement of the time block.

Up to its final mesh interval, the polygonal path lies in the
$\ell_\lambda$-neighborhood of $\overline{D_\lambda}$. Applying
$\mathcal R_{\lambda,\ell_\lambda}$ pushes this portion into
$D_\lambda$ and changes its action by $o(1)$ according to
\cref{eq:exterior-ball-radial-deformation}. At the end of the
deformed portion, append a radial segment of length
$O(\ell_\lambda)$ which ends on $\partial D_\lambda$. Traverse this
segment in time $O(\ell_\lambda)$. Since $V$ is bounded on the fixed
compact set under consideration, its kinetic and potential actions
are both $O(\ell_\lambda)$. Thus the resulting path is admissible for
the enlarged space-time block and its action differs from the
polygonal action by $o(1)$.

This proves the upper bound with the infimum taken over the enlarged
block. The enlargement does not affect the summed estimate. Indeed,
every enlarged time interval or boundary patch meets only a bounded
number of the original overlapping blocks. Reassigning its
contribution to these neighboring blocks therefore changes the sum
by only a fixed factor. Moreover, the endpoint profiles used below
are uniformly locally Lipschitz, so moving exit data by
$O(\ell_\lambda)$ changes those profiles by $o(1)$. Since the number
of blocks is $\exp(o(\lambda))$, the full upper bound incurs only an
$\exp(o(\lambda))$ loss.

This argument also covers exits occurring during the first fine mesh
interval. For completeness, if the starting point has distance at
least $\ell_\lambda$ from the boundary, the Gaussian maximal estimate
gives
\eq{
\PP_z\left[
\tau_{D_\lambda}
\leq
\frac{a_\lambda}{\lambda}
\right]
\leq
C
\exp\left(
-\frac{
c\lambda\ell_\lambda^2
}{
a_\lambda
}
\right),
}
which is superexponentially small by
\cref{eq:exterior-ball-mesh-scales}. Starting points in the
$\ell_\lambda$ boundary layer are retained in the block decomposition
and are handled by the preceding deformation argument.

We now prove the lower bound. Let $\gamma$ be an admissible path whose
action is within $\varepsilon$ of the constrained infimum for a
space-time block. By an arbitrarily small perturbation in a terminal
neighborhood, we may arrange that $\gamma$ remains strictly inside
$D_\lambda$ before its terminal time and reaches
$\partial D_\lambda$ transversely at a point in the interior of the
chosen boundary patch and at a time in the interior of the chosen
time block. This perturbation changes the action by at most
$\varepsilon$. Extend $\gamma$ for a short time after its terminal
time along the outward normal. The extension may be chosen with
bounded speed and duration $O(\ell_\lambda)$, and hence has kinetic
action $O(\ell_\lambda)$.

For all sufficiently large $\lambda$, the portion of $\gamma$ before
the final $O(\ell_\lambda)$ time interval has a positive distance
from the boundary. Push that portion farther into the domain, if
necessary, by applying
$\mathcal R_{\lambda,3\ell_\lambda}$. In the final
$O(\ell_\lambda)$ interval, retain the transverse terminal segment
and its outward extension. We thereby obtain a comparison path which
has distance at least $2\ell_\lambda$ from the boundary before its
terminal interval, crosses the boundary transversely inside the
prescribed block, and lies at distance at least
$2\ell_\lambda$ outside the domain at the end of its outward
extension. Its action before the exit differs from that of $\gamma$
by $o(1)$.

Restrict every fine-mesh value before the outward extension to a cube
of side $b_\lambda$ about this comparison path and restrict each
conditional Brownian bridge to remain within $\ell_\lambda$ of its
linear interpolation. Before the terminal interval, every path in
this event remains inside $D_\lambda$. At the end of the outward
extension it lies outside $D_\lambda$. Continuity therefore forces a
first exit during the prescribed enlarged time block and through the
prescribed enlarged boundary patch. The bridge restriction has
conditional probability
\eq{
1-
O\left(
m
\exp\left(
-\frac{
c\lambda\ell_\lambda^2
}{
a_\lambda
}
\right)
\right)
=
1-o(1),
}
while the mesh-cube volumes and Gaussian normalization factors
contribute only $\exp(o(\lambda))$. The additional kinetic action of
the outward extension is $o(1)$, and the Feynman--Kac potential is
integrated only up to the first exit. Letting
$\varepsilon\downarrow0$ proves the matching lower bound.

We have therefore proved the first-exit Laplace bounds uniformly over
the space-time blocks. The exit blocks have diameter
$O(\ell_\lambda)$, while the logarithmic endpoint profiles used below
are uniformly locally Lipschitz. Replacing a point by any other point
in the same or a neighboring block consequently changes those
profiles by $o(1)$. Thus all blockwise sums retain only
$\exp(o(\lambda))$ losses.

We also record the total-exit upper bound used in
\cref{lem:integrated-defect-estimate}. Set
\eql{
u_\lambda(x)
:=
\EE_x\left[
\exp\left(
-\lambda^2\int_0^{\tau_{D_\lambda}}V(\gamma(t))\dif{t}
\right);
\ \tau_{D_\lambda}<\infty
\right].
\label{eq:total-exit-harmonic-function}
}
Thus $u_\lambda$ is the minimal positive solution of
\eq{
\br{-\Delta+\lambda^2V}u_\lambda=0
\quad\text{in }D_\lambda,
\qquad
u_\lambda=1
\quad\text{on }\partial D_\lambda.
}
Write $H_\lambda^{D_\lambda}$ for the Dirichlet realization of
$-\Delta+\lambda^2V$. Choose
$g_\lambda\in C_c^\infty(\overline{D_\lambda})$ with boundary trace
one, supported in the collar
\eq{
\dist\br{x,\partial D_\lambda}
\leq
2\lambda^{-1/2},
}
and with
\eq{
\norm{\nabla^j g_\lambda}_\infty
\leq
C\lambda^{j/2},
\qquad
j\in\Set{0,1,2}.
}
The collar is contained in an $O(r_\lambda)$ neighborhood of the
excluded minimum. Taylor's theorem and radial comparison with the
boundary therefore give
\eql{
\sup_{\supp g_\lambda}V
\leq
C\lambda^{-1/2},
\qquad
\sup_{\supp g_\lambda}
S_{D_\lambda}\br{\cdot,\partial D_\lambda}
\leq
C\lambda^{-3/4}.
\label{eq:total-exit-collar-bounds}
}

By time reparametrization, $S_{D_\lambda}$ is the intrinsic distance
for the length element $\sqrt V\,\abs{\dif{x}}$. Consequently,
$S_{D_\lambda}(\cdot,\partial D_\lambda)$ is locally Lipschitz and
\eq{
\norm{
\nabla S_{D_\lambda}\br{\cdot,\partial D_\lambda}
}^2
\leq
V
}
almost everywhere. Put
\eq{
\Phi_\lambda
:=
\br{1-\lambda^{-1/2}}
\lambda
S_{D_\lambda}\br{\cdot,\partial D_\lambda}.
}
For $R>0$, put $\Phi_{\lambda,R}:=\Phi_\lambda\wedge R$. Initially
for $f\in C_c^\infty(D_\lambda)$, and hence by closure in the
quadratic-form sense, the Agmon identity gives
\eql{
\operatorname{Re}
\left\langle
f,
\ee^{\Phi_{\lambda,R}}
H_\lambda^{D_\lambda}
\ee^{-\Phi_{\lambda,R}}f
\right\rangle
&=
\norm{\nabla f}_2^2
+
\int_{D_\lambda}
\left(
\lambda^2V-\norm{\nabla\Phi_{\lambda,R}}^2
\right)
\abs{f}^2\dif{x}
\nonumber\\
&\geq
\left[
1-\br{1-\lambda^{-1/2}}^2
\right]
\left\langle
f,H_\lambda^{D_\lambda}f
\right\rangle
\nonumber\\
&\geq
c\lambda^{1/2}\norm{f}_2^2.
\label{eq:total-exit-weighted-coercivity}
}
Here the last inequality is precisely the assumed bound
$e_{\lambda,0}^{D_\lambda}\geq c\lambda$; in particular, the zero of
$V$ at the retained minimum causes no loss of coercivity. Applying
this form estimate to the weak resolvent equation, with compact
spatial cutoffs and then by density, gives, uniformly in $R$,
\eql{
\sup_{R>0}
\norm{
\ee^{\Phi_{\lambda,R}}
\br{H_\lambda^{D_\lambda}}^{-1}
\ee^{-\Phi_{\lambda,R}}
}_{2\to2}
\leq
C\lambda^{-1/2}.
\label{eq:total-exit-weighted-resolvent}
}
The stopped Feynman--Kac formula, followed by Dirichlet exhaustion of
the exterior domain, gives the exact identity
\eql{
u_\lambda
=
g_\lambda
-
\br{H_\lambda^{D_\lambda}}^{-1}H_\lambda g_\lambda.
\label{eq:total-exit-resolvent-representation}
}
Indeed, $H_\lambda^{D_\lambda}$ is invertible by the same spectral
lower bound, and this identity selects the minimal solution at
infinity; here $H_\lambda g_\lambda$ denotes the differential
expression, which is compactly supported and belongs to $L^2$. The
derivative bounds on $g_\lambda$ and
\cref{eq:total-exit-collar-bounds} imply
\eq{
\norm{\ee^{\Phi_\lambda}g_\lambda}_2
+
\norm{\ee^{\Phi_\lambda}H_\lambda g_\lambda}_2
\leq
\ee^{o(\lambda)}.
}
Thus the uniform truncated estimate in
\cref{eq:total-exit-weighted-resolvent,eq:total-exit-resolvent-representation}
and Fatou's lemma give
\eql{
\norm{\ee^{\Phi_\lambda}u_\lambda}_2
\leq
\ee^{o(\lambda)}.
\label{eq:total-exit-weighted-L2}
}

The compact starting families used below remain a fixed positive
distance from $\partial D_\lambda$. Since
$\Delta u_\lambda=\lambda^2Vu_\lambda\geq0$, the mean-value
inequality on $B_{\lambda^{-1}}(x)$, together with
\cref{eq:total-exit-weighted-L2} and the local Lipschitz bound on
$\Phi_\lambda$, gives
\eq{
u_\lambda(x)
\leq
C\lambda^{n/2}
\norm{u_\lambda}_{L^2(B_{\lambda^{-1}}(x))}
\leq
\exp\left(
-\lambda S_{D_\lambda}\br{x,\partial D_\lambda}
+o(\lambda)
\right).
}
Here $S_{D_\lambda}(x,\partial D_\lambda)$ is uniformly bounded on
the starting families, so replacing $\Phi_\lambda(x)$ by
$\lambda S_{D_\lambda}(x,\partial D_\lambda)$ costs only
$o(\lambda)$.
This proves, uniformly on those families,
\eql{
\EE_x\left[
\exp\left(
-\lambda^2\int_0^{\tau_{D_\lambda}}V(\gamma(t))\dif{t}
\right);
\ \tau_{D_\lambda}<\infty
\right]
\leq
\exp\left(
-\lambda S_{D_\lambda}\br{x,\partial D_\lambda}+o(\lambda)
\right).
\label{eq:total-exit-agmon-resolvent-bound}
}

It remains to pass from the total exit to the pointwise first-exit
density. Let $\nu_{D_\lambda}$ denote the outward unit normal and
write
\eq{
J_\lambda^{D_\lambda}(t;x,y)
:=
-\nu_{D_\lambda}(y)\cdot
\nabla_yK_{\lambda,t}^{D_\lambda}(x,y),
\qquad
y\in\partial D_\lambda.
}
Feynman--Kac domination, followed by the inward boundary quotient,
bounds this by the corresponding Dirichlet first-exit density
$J_0^{D_\lambda}$ with $V=0$. The free Gaussian bound followed by the
local Dirichlet boundary gradient estimate on a parabolic cylinder of
spatial radius $\sqrt{s}$ gives
\eql{
0
\leq
J_\lambda^{D_\lambda}(s;z,y)
\leq
J_0^{D_\lambda}(s;z,y)
\leq
Cs^{-(n+1)/2}
\exp\left(
-\frac{c\norm{z-y}^2}{s}
\right),
\qquad
0<s\leq r_\lambda^2.
\label{eq:exterior-ball-poisson-kernel-bound}
}
The constants are uniform in $\lambda$: after rescaling by
$\sqrt{s}$, the curvature of the boundary sphere is bounded by
$\sqrt{s}/r_\lambda\leq1$.

Set $\tau_\lambda:=\lambda^{-2}$. For every fixed $R>0$,
\eql{
\inf_{\substack{
y\in\partial D_\lambda,\ z\in D_\lambda\\
\norm{z-y}\leq R/\sqrt\lambda
}}
u_\lambda(z)
\geq
\exp\br{-C_R\sqrt\lambda}.
\label{eq:total-exit-boundary-layer-lower}
}
Indeed, constrain the Brownian motion for a time $O_R(\lambda^{-1})$
to a tube of radius $O_R(\lambda^{-1/2})$ which reaches the boundary.
After parabolic rescaling, this event has probability bounded below
by a positive constant depending only on $R$. Throughout the tube,
$V\leq Cr_\lambda^2=C\lambda^{-1/2}$, so its Feynman--Kac weight is
at least $\exp\br{-C_R\sqrt\lambda}$.

For $t>\tau_\lambda$, the semigroup property and differentiation at
the boundary give
\eql{
J_\lambda^{D_\lambda}(t;x,y)
=
\int_{z\in D_\lambda}
K_{\lambda,t-\tau_\lambda}^{D_\lambda}(x,z)
J_\lambda^{D_\lambda}(\tau_\lambda;z,y)
\dif{z}.
\label{eq:pointwise-exit-semigroup}
}
Split this integral according to
$\norm{z-y}\leq R/\sqrt\lambda$. By the strong Markov property and
sub-Markovianity,
\eq{
\int_{z\in D_\lambda}
K_{\lambda,s}^{D_\lambda}(x,z)u_\lambda(z)\dif{z}
\leq
u_\lambda(x),
\qquad
\int_{z\in D_\lambda}
K_{\lambda,s}^{D_\lambda}(x,z)\dif{z}
\leq
1.
}
Consequently,
\cref{eq:exterior-ball-poisson-kernel-bound,eq:total-exit-boundary-layer-lower,eq:pointwise-exit-semigroup}
give
\eql{
J_\lambda^{D_\lambda}(t;x,y)
\leq
C\lambda^{n+1}
\left[
\ee^{C_R\sqrt\lambda}u_\lambda(x)
+
\ee^{-cR^2\lambda}
\right],
\qquad
t>\tau_\lambda.
\label{eq:pointwise-exit-from-total-exit}
}
The actions are uniformly bounded on every fixed compact family of
starting points. We may therefore fix $R$ so large that
\eq{
cR^2
>
1+
\sup_{\lambda,x}
S_{D_\lambda}\br{x,\partial D_\lambda},
}
where the supremum is restricted to the starting family. The second
term in \cref{eq:pointwise-exit-from-total-exit} then has the required
exponential bound. It follows that
\eq{
J_\lambda^{D_\lambda}(t;x,y)
\leq
\exp\left(
-\lambda S_{D_\lambda}\br{x,\partial D_\lambda}
+o(\lambda)
\right),
\qquad
t>\tau_\lambda.
}
For the starting families used below there is also a fixed positive
lower bound on $\dist(x,\partial D_\lambda)$. Hence, for
$0<t\leq\tau_\lambda$,
\cref{eq:exterior-ball-poisson-kernel-bound} gives instead
\eq{
J_\lambda^{D_\lambda}(t;x,y)
\leq
\exp\br{-c\lambda^2}.
}
Since $\abs{\partial D_\lambda}$ is polynomial in $\lambda$, we have
therefore proved, for every $T_\lambda\downarrow0$,
\eql{
t\abs{\partial D_\lambda}
J_\lambda^{D_\lambda}(t;x,y)
\leq
\exp\left(
-\lambda S_{D_\lambda}\br{x,\partial D_\lambda}
+o(\lambda)
\right),
\qquad
0<t\leq T_\lambda,
\label{eq:pointwise-first-exit-bound}
}
uniformly in the starting and exit points under consideration.

We next prove
\cref{eq:exterior-ball-ground-state-profile}. Write
\eq{
\vf_\lambda
:=
\vf_{\lambda,0}^{D_\lambda},
\qquad
e_\lambda
:=
e_{\lambda,0}^{D_\lambda},
\qquad
S_\lambda(x)
:=
S_{D_\lambda}(q,x).
}
Domination by the free heat kernel at time $1/\lambda$ gives
\eq{
\norm{\vf_\lambda}_\infty
\leq
\ee^{e_\lambda/\lambda}
\norm{
\ee^{-H_\lambda^{D_\lambda}/\lambda}
}_{2\to\infty}
\leq
C\lambda^{n/4}.
}

The nondegeneracy of the two minima and
$r_\lambda=\lambda^{-1/4}$ imply
\eql{
\lambda
\inf_{
x\in
D_\lambda
\setminus
B_{\lambda^{-1/2+\eta}}(q)
}
V(x)
\geq
c
\min\Set{
\lambda^{2\eta},
\lambda^{1/2},
\lambda
}
\longrightarrow
\infty.
\label{eq:exterior-ball-coercivity}
}
The three terms in
\cref{eq:exterior-ball-coercivity} correspond, respectively, to the
neighborhood of $q$, the neighborhood of the excluded well, and the
complement of fixed neighborhoods of the two wells.

We shall also use the Agmon-length representation
\eql{
S_\lambda(x)
=
\inf_{\substack{
\sigma(0)=q,\ 
\sigma(1)=x\\
\sigma((0,1))
\subseteq
D_\lambda
}}
\int_0^1
\sqrt{
V\br{\sigma(s)}
}
\norm{\dot\sigma(s)}
\dif{s}.
\label{eq:exterior-ball-agmon-length}
}
For one direction, use
\eq{
\frac14
\norm{\dot\gamma}^2
+
V(\gamma)
\geq
\sqrt{V(\gamma)}
\norm{\dot\gamma}.
}
For the reverse direction, parameterize a spatial path away from the
zeros of $V$ by
\eq{
\norm{\dot\gamma}
=
2\sqrt{V(\gamma)}.
}
Near a zero of $V$, replace this parameterization by constant speed
and then shrink the replaced neighborhood. Since
$V(x)\leq C\norm{x-q}^2$ near $q$, the resulting excess action tends
to zero. This proves \cref{eq:exterior-ball-agmon-length}.

It follows from \cref{eq:exterior-ball-agmon-length} that $S_\lambda$
is locally Lipschitz and
\eq{
\norm{\nabla S_\lambda(x)}^2
\leq
V(x)
}
almost everywhere. Fix $\varepsilon\in(0,1)$ and apply the Agmon
identity to the Dirichlet ground state with the Lipschitz weight
\eq{
\exp\br{
(1-\varepsilon)
\lambda S_\lambda
}.
}
Using bounded smooth approximations to the weight, one obtains
\eql{
0
&=
\int_{x\in D_\lambda}
\norm{
\nabla\left(
\ee^{(1-\varepsilon)\lambda S_\lambda(x)}
\vf_\lambda(x)
\right)
}^2
\dif{x}
\nonumber\\
&\quad+
\int_{x\in D_\lambda}
\left[
\lambda^2V(x)
-
e_\lambda
-
(1-\varepsilon)^2
\lambda^2
\norm{\nabla S_\lambda(x)}^2
\right]
\nonumber\\
&\hspace{8em}\times
\ee^{2(1-\varepsilon)\lambda S_\lambda(x)}
\vf_\lambda(x)^2
\dif{x}.
\label{eq:exterior-ball-agmon-identity}
}
The expression in square brackets is bounded below by
\eq{
\br{
1-(1-\varepsilon)^2
}
\lambda^2V(x)
-
e_\lambda.
}
By \cref{eq:exterior-ball-coercivity} and $e_\lambda\leq C\lambda$,
this expression is positive outside
$B_{\lambda^{-1/2+\eta}}(q)$ for all sufficiently large $\lambda$.

Inside that ball, the straight path from $q$ to $x$ and the
nondegeneracy of the minimum give
\eq{
\sup_{
x\in
B_{\lambda^{-1/2+\eta}}(q)
}
\lambda S_\lambda(x)
=
O\br{\lambda^{2\eta}}
=
o(\lambda).
}
Using the $L^2$ normalization of $\vf_\lambda$ in
\cref{eq:exterior-ball-agmon-identity} therefore gives
\eql{
\int_{x\in D_\lambda}
\ee^{2(1-\varepsilon)\lambda S_\lambda(x)}
\vf_\lambda(x)^2
\dif{x}
\leq
\exp(o(\lambda)).
\label{eq:exterior-ball-weighted-L2}
}
Interior and Dirichlet boundary elliptic estimates on balls and
boundary half-balls of radius $\lambda^{-1}$ turn
\cref{eq:exterior-ball-weighted-L2} into
\eq{
\limsup_{\lambda\to\infty}
\sup_{x\in K\cap D_\lambda}
\left\{
\frac1\lambda
\log\vf_\lambda(x)
+
(1-\varepsilon)
S_\lambda(x)
\right\}
\leq0.
}
These estimates are uniform at the boundary because
\eq{
\lambda^{-1}
\ll
r_\lambda,
\qquad
\lambda r_\lambda
\longrightarrow
\infty,
}
so the boundary curvature is uniformly bounded after rescaling by
$\lambda$. Since $S_\lambda$ is uniformly bounded on
$K\cap D_\lambda$, letting $\varepsilon\downarrow0$ proves
\eq{
\limsup_{\lambda\to\infty}
\sup_{x\in K\cap D_\lambda}
\left\{
\frac1\lambda
\log\vf_\lambda(x)
+
S_\lambda(x)
\right\}
\leq0.
}

For the reverse inequality, we first prove the uniform finite-time
approximation
\eql{
S_{T_{\varepsilon,K},D_\lambda}(q,x)
\leq
S_{D_\lambda}(q,x)
+
\varepsilon
\label{eq:uniform-finite-time-approximation}
}
for some $T_{\varepsilon,K}<\infty$ independent of $\lambda$, uniformly
for
\eq{
x\in K\cap D_\lambda,
\qquad
\dist(x,\partial D_\lambda)
\geq
r_\lambda.
}

To prove
\cref{eq:uniform-finite-time-approximation}, note first that the
exterior of a ball is uniformly quasiconvex. If a straight segment
meets the excluded ball, replace the part inside that ball by the
shorter spherical arc. The resulting path has length bounded by a
constant multiple of the distance between its endpoints, uniformly
in $\lambda$. The corresponding statement in dimension one is
immediate.

Choose an almost minimizing path in
\cref{eq:exterior-ball-agmon-length}. Outside fixed small
neighborhoods of the two minima, $V$ is bounded below by a positive
constant. Its Euclidean length there may therefore be taken uniformly
bounded. Inside a $\delta$-neighborhood of either minimum, replace
the path by radial segments and, near the excluded well, by a
spherical arc if necessary. The replacement has Euclidean length
$O(\delta)$ and changes the Agmon length by at most $O(\delta^2)$.

On the part where $V$ is bounded below, parameterize the resulting
path by
\eq{
\norm{\dot\gamma}
=
2\sqrt{V(\gamma)}.
}
Traverse the portions inside the $\delta$-neighborhoods at constant
speed in a fixed amount of time. Their kinetic and potential actions
are both $O(\delta^2)$. Thus the total traversal time is bounded
independently of $\lambda$ and $x$, while the action exceeds
$S_{D_\lambda}(q,x)$ by at most $C\delta^2$. Choose $\delta$ so that
$C\delta^2<\varepsilon$. Waiting at $q$, where $V(q)=0$, makes the
total time equal to one common value $T_{\varepsilon,K}$. This proves
\cref{eq:uniform-finite-time-approximation}.

The localization hypothesis gives points
\eq{
w_\lambda
\in
B_{\lambda^{-1/2+\eta}}(q)
}
such that
\eq{
\vf_\lambda(w_\lambda)
\geq
\exp\br{
-o(\lambda)
}.
}
After rescaling the ground-state equation on a ball of radius
$c/\lambda$ about $w_\lambda$, the zeroth-order coefficient remains
uniformly bounded. The interior Harnack inequality therefore gives a
ball
\eq{
B_{c/\lambda}(w_\lambda)
\subseteq
D_\lambda
}
on which
\eql{
\inf_{
z\in
B_{c/\lambda}(w_\lambda)
}
\vf_\lambda(z)
\geq
\exp\br{
-o(\lambda)
}.
\label{eq:exterior-ball-local-ground-state-lower-bound}
}
Moreover,
\eq{
S_{D_\lambda}(q,z)
=
o(1)
}
uniformly on this ball.

Set
\eq{
\widehat T_{\varepsilon,K}
:=
T_{\varepsilon,K}+1.
}
The ground-state equation gives
\eq{
\vf_\lambda(x)
=
\ee^{
\widehat T_{\varepsilon,K}e_\lambda/\lambda
}
\int_{z\in D_\lambda}
K_{\lambda,
\widehat T_{\varepsilon,K}/\lambda}^{D_\lambda}(x,z)
\vf_\lambda(z)
\dif{z}.
}
Restrict the integral to
$B_{c/\lambda}(w_\lambda)$. Reverse a path supplied by
\cref{eq:uniform-finite-time-approximation} and use the remaining
unit of time to join $q$ to $z$. The latter part has action $o(1)$.
Consequently,
\eq{
S_{\widehat T_{\varepsilon,K},D_\lambda}(x,z)
\leq
S_{D_\lambda}(q,x)
+
\varepsilon
+
o(1).
}
Using \cref{eq:exterior-ball-kernel} and
\cref{eq:exterior-ball-local-ground-state-lower-bound}, and observing
that the volume of the integration ball is subexponential, gives
\eq{
-\frac1\lambda
\log
\vf_\lambda(x)
\leq
S_{D_\lambda}(q,x)
+
\varepsilon
+
o(1)
}
uniformly for
\eq{
x\in K\cap D_\lambda,
\qquad
\dist(x,\partial D_\lambda)
\geq
r_\lambda.
}
Letting $\varepsilon\downarrow0$ proves
\cref{eq:exterior-ball-ground-state-profile}. The same proof is
uniform on the stated hypersurfaces.

It remains to prove
\cref{eq:exterior-ball-integrated-exit}. Denote its integral by
$I_\lambda(L)$. By
\cref{eq:exterior-ball-ground-state-profile}, the initial and terminal
weights have the logarithmic profiles
\eq{
S_{\Omega_\lambda}(d,z)
\qquad\text{and}\qquad
S_{-\Omega_\lambda}(y,-d),
}
respectively. For the initial factor, the corresponding upper bound
remains valid in the boundary layer. For the terminal factor, note
that
\eq{
\dist\br{
\partial\Omega_\lambda,
\partial(-\Omega_\lambda)
}
\geq
c>0
}
for all sufficiently large $\lambda$, so
\cref{eq:exterior-ball-ground-state-profile} applies uniformly on
$\partial\Omega_\lambda$.

Both profiles are uniformly locally Lipschitz on the relevant compact
sets. Indeed, nearby points in the exterior of a ball can be joined by
a path whose length is at most a constant multiple of their Euclidean
distance, and $V$ is bounded on these compact sets. Hence
\eq{
\abs{
S_{-\Omega_\lambda}(y,-d)
-
S_{-\Omega_\lambda}(y',-d)
}
\leq
C\norm{y-y'}
}
for nearby $y,y'\in\partial\Omega_\lambda$, and the analogous estimate
holds for the initial profile.

The $z$-integration is exponentially tight. By the polynomial
$L^\infty$ bounds already proved and the Gaussian maximal estimate,
\eq{
\PP_z\left[
\tau_{\Omega_\lambda}
\leq
L/\lambda
\right]
\leq
C
\exp\left(
-\frac{c\lambda}{L}
\dist\br{
z,
\partial\Omega_\lambda
}^2
\right).
}
Since $\partial\Omega_\lambda$ remains in a fixed compact set, the
contribution of $z\notin B_R(0)$ can be made smaller than
$\exp(-M\lambda)$, for any prescribed $M$, by choosing $R$ large.

Restrict now to $B_R(0)\cap\Omega_\lambda$ and cover the starting
points by cubes of side $b_\lambda$. Cover $(0,L]$ by the overlapping
exit-time blocks of length $O(\ell_\lambda)$ constructed above, and
cover $\partial\Omega_\lambda$ by the corresponding boundary patches
of diameter $O(\ell_\lambda)$. The total number of resulting blocks is
$\exp(o(\lambda))$. The first-exit upper bound proved above, together
with the upper logarithmic profiles of the two ground states, gives
\eq{
\liminf_{\lambda\to\infty}
\left(
-\frac1\lambda
\log
I_\lambda(L)
\right)
\geq
\liminf_{\lambda\to\infty}
\inf_{\substack{
z\in\Omega_\lambda,\ 
y\in\partial\Omega_\lambda\\
0<T\leq L
}}
\left\{
S_{\Omega_\lambda}(d,z)
+
S_{T,\Omega_\lambda}(z,y)
+
S_{-\Omega_\lambda}(y,-d)
\right\}.
}
The starting-point cover includes all cubes meeting the
$O(r_\lambda)$ boundary layer, while the overlapping exit-time cover
includes times arbitrarily close to zero. Thus no potentially leading
boundary-layer or short-time contribution has been removed.

For the reverse bound, set
\eq{
J_\lambda
:=
\inf_{y\in\partial\Omega_\lambda}
\left\{
S_{\Omega_\lambda}(d,y)
+
S_{-\Omega_\lambda}(y,-d)
\right\}.
}
Concatenation gives
\eq{
S_{\Omega_\lambda}(d,y)
\leq
S_{\Omega_\lambda}(d,z)
+
S_{T,\Omega_\lambda}(z,y),
}
and hence
\eq{
\inf_{\substack{
z\in\Omega_\lambda,\ 
y\in\partial\Omega_\lambda\\
0<T\leq L
}}
\left\{
S_{\Omega_\lambda}(d,z)
+
S_{T,\Omega_\lambda}(z,y)
+
S_{-\Omega_\lambda}(y,-d)
\right\}
\geq
J_\lambda.
}
Conversely, splitting an almost minimizing path sufficiently close to
its terminal point shows that, for every
$y\in\partial\Omega_\lambda$,
\eq{
\inf_{\substack{
z\in\Omega_\lambda\\
0<T\leq L
}}
\left\{
S_{\Omega_\lambda}(d,z)
+
S_{T,\Omega_\lambda}(z,y)
\right\}
=
S_{\Omega_\lambda}(d,y).
}
Therefore the variational expression in
\cref{eq:exterior-ball-integrated-exit} equals $J_\lambda$.

Choose $y_\lambda\in\partial\Omega_\lambda$ such that
\eq{
S_{\Omega_\lambda}(d,y_\lambda)
+
S_{-\Omega_\lambda}(y_\lambda,-d)
\leq
J_\lambda
+
o(1).
}
Write
\eq{
y_\lambda
=
-d
+
r_\lambda\omega_\lambda
}
and define
\eq{
z_\lambda
:=
-d
+
3r_\lambda\omega_\lambda.
}
Then
\eq{
\dist\br{
z_\lambda,
\partial\Omega_\lambda
}
=
2r_\lambda.
}
The radial segment from $z_\lambda$ to $y_\lambda$, traversed in any
fixed time $T_*\in(0,L)$, has action
\eql{
S_{T_*,\Omega_\lambda}(z_\lambda,y_\lambda)
\leq
C_Lr_\lambda^2
=
o(1).
\label{eq:exterior-ball-short-exit-segment}
}
Indeed, its kinetic action is $O(r_\lambda^2)$ and
$V(x)\leq C\norm{x+d}^2=O(r_\lambda^2)$ along the segment. Similarly,
\cref{eq:exterior-ball-agmon-length} gives
\eql{
S_{\Omega_\lambda}(d,z_\lambda)
\leq
S_{\Omega_\lambda}(d,y_\lambda)
+
Cr_\lambda^2.
\label{eq:exterior-ball-boundary-layer-profile}
}

Restrict the $z$-integral to a cube of side $b_\lambda$ about
$z_\lambda$, the exit time to an
$O(\ell_\lambda)$ block containing $T_*$ in its interior, and the exit
point to an $O(\ell_\lambda)$ boundary patch containing
$y_\lambda$ in its interior. The starting cube remains at distance at
least $r_\lambda$ from the boundary. Hence
\cref{eq:exterior-ball-ground-state-profile} supplies the lower bound
for the initial ground state there, while the same estimate for
$-\Omega_\lambda$ supplies the lower bound for the terminal ground
state.

Alter the radial path from $z_\lambda$ to $y_\lambda$ by an
$o(1)$-action perturbation so that it reaches the boundary transversely
at time $T_*$. Extend it a distance $O(\ell_\lambda)$ along the
outward normal after that time. The first-exit lower bound proved above
then applies to this block. Together with
\cref{eq:exterior-ball-short-exit-segment,eq:exterior-ball-boundary-layer-profile},
it gives
\eq{
\limsup_{\lambda\to\infty}
\left(
-\frac1\lambda
\log
I_\lambda(L)
\right)
\leq
\limsup_{\lambda\to\infty}
J_\lambda.
}
This proves
\cref{eq:exterior-ball-integrated-exit}. The same blockwise argument
gives the stated uniform upper bounds and the lower bounds for interior
starting blocks.
\end{proof}

\section{Orthogonal-complement estimates for the one-passage sector}
\label[appendix]{app:one-passage-error-estimates}

We give the details of the error estimate used in
\cref{eq:one passage flux asymptotic}. Throughout this appendix,
$\lambda$ is fixed, so all constants implicit in $\calO_\lambda$ may depend
on $\lambda$ but not on $\beta$.

Using the decomposition
$\Id=P_{\lambda,0}^{\Omega_\lambda}+Q_\lambda^{\Omega_\lambda}$ in the
first heat kernel in \cref{eq:one passage flux convolution}, we have
\eq{
-\nu_\lambda(y)\cdot\nabla_y
K_{\lambda,t}^{\Omega_\lambda}(x,y)
&=
\ee^{-t e_\lambda^D}
\vf_{\lambda,0}^{\Omega_\lambda}(x)
\left[
-\nu_\lambda(y)\cdot
\nabla\vf_{\lambda,0}^{\Omega_\lambda}(y)
\right]
\\
&\qquad-
\nu_\lambda(y)\cdot\nabla_y
\left(
\ee^{-t h_\lambda^{\Omega_\lambda}}
Q_\lambda^{\Omega_\lambda}
\right)(x,y).
}
By reflection symmetry, the corresponding decomposition of the second
heat kernel is
\eq{
K_{\lambda,t}^{-\Omega_\lambda}(y,-x)
&=
\ee^{-t e_\lambda^D}
\vf_{\lambda,0}^{-\Omega_\lambda}(y)
\vf_{\lambda,0}^{-\Omega_\lambda}(-x)
\\
&\qquad+
\left(
\ee^{-t h_\lambda^{-\Omega_\lambda}}
R Q_\lambda^{\Omega_\lambda}R
\right)(y,-x).
}
Thus the expansion of \cref{eq:one passage flux convolution} has the
four contributions $PP$, $QP$, $PQ$, and $QQ$.

We first record the estimates on the two factors containing $Q$. For
fixed $\lambda$ there is a finite constant, independent of $t$, such
that
\begin{malign}
&\left(
\int_{x\in A_{\lambda,+}}
\chi_{\lambda,+}(x)^2
\int_{y\in\partial\Omega_\lambda}
\left|
\nu_\lambda(y)\cdot\nabla_y
\left(
\ee^{-t h_\lambda^{\Omega_\lambda}}
Q_\lambda^{\Omega_\lambda}
\right)(x,y)
\right|^2
\dif{y}\dif{x}
\right)^{1/2}
\\
&\qquad\leq
\calO_\lambda(1)
\ee^{-t\br{e_\lambda^D+\delta_\lambda^D}},
\end{malign}
and
\begin{malign}
&\left(
\int_{x\in A_{\lambda,+}}
\chi_{\lambda,+}(x)^2
\int_{y\in\partial\Omega_\lambda}
\left|
\left(
\ee^{-t h_\lambda^{-\Omega_\lambda}}
R Q_\lambda^{\Omega_\lambda}R
\right)(y,-x)
\right|^2
\dif{y}\dif{x}
\right)^{1/2}
\\
&\qquad\leq
\calO_\lambda(1)
\ee^{-t\br{e_\lambda^D+\delta_\lambda^D}}.
\end{malign}
Here is the standard justification, including the normal derivative in
the first estimate. For $t\geq2$,
\eq{
\ee^{-t h_\lambda^{\Omega_\lambda}}
Q_\lambda^{\Omega_\lambda}
=
\ee^{-h_\lambda^{\Omega_\lambda}}
\left(
\ee^{-(t-2)h_\lambda^{\Omega_\lambda}}
Q_\lambda^{\Omega_\lambda}
\right)
\ee^{-h_\lambda^{\Omega_\lambda}}.
}
The two fixed-time heat operators provide the required smoothing: one
controls evaluation on the compact support of $\chi_{\lambda,+}$ and
the other controls the Dirichlet boundary normal trace on
$\partial\Omega_\lambda$. Standard parabolic regularity and the spectral
theorem therefore give
\eq{
\calO_\lambda(1)
\left\|
\ee^{-(t-2)h_\lambda^{\Omega_\lambda}}
Q_\lambda^{\Omega_\lambda}
\right\|_{2\to2}
&\leq
\calO_\lambda(1)
\ee^{-(t-2)\br{e_\lambda^D+\delta_\lambda^D}}
\\
&=
\calO_\lambda(1)
\ee^{-t\br{e_\lambda^D+\delta_\lambda^D}}.
}
The reflected argument gives the second estimate. For $0<t\leq2$, the
support of $\chi_{\lambda,+}$ is at positive distance from
$\partial\Omega_\lambda$, and its reflection is likewise at positive
distance from that hypersurface. The standard off-diagonal short-time
heat-kernel estimate and its Dirichlet boundary normal-derivative version
give
\eq{
\left|
\nu_\lambda(y)\cdot\nabla_y
K_{\lambda,t}^{\Omega_\lambda}(x,y)
\right|
+
\left|
K_{\lambda,t}^{-\Omega_\lambda}(y,-x)
\right|
\leq
\calO_\lambda\left(
 t^{-M}\ee^{-c_\lambda/t}
\right)
}
for some $M<\infty$ on the sets occurring above. After subtracting the
smooth ground-state terms, the two displayed $Q$-expressions are
therefore $\calO_\lambda(1)$ for $0<t\leq2$. Since
$\ee^{-t\br{e_\lambda^D+\delta_\lambda^D}}$ is bounded below by a
positive constant on this interval, the preceding two estimates hold for
all $t>0$. This is the reason one separates short endpoint times from the
intervening region when deriving the estimates directly from the
spectral gap.

We now bound the three contributions containing at least one $Q$.
First, the contribution with $Q$ in the first heat kernel and the
ground-state projection in the second is bounded, by Cauchy--Schwarz in
$(x,y)$, by
\begin{malign}
&2\left|
\int_0^\beta
\int_{x\in A_{\lambda,+}}
\chi_{\lambda,+}(x)^2
\int_{y\in\partial\Omega_\lambda}
\left[
-\nu_\lambda(y)\cdot\nabla_y
\left(
\ee^{-t h_\lambda^{\Omega_\lambda}}
Q_\lambda^{\Omega_\lambda}
\right)(x,y)
\right]
\right.
\\
&\hspace{8em}\left.
\times
\ee^{-(\beta-t)e_\lambda^D}
\vf_{\lambda,0}^{-\Omega_\lambda}(y)
\vf_{\lambda,0}^{-\Omega_\lambda}(-x)
\dif{y}\dif{x}\dif{t}
\right|
\\
&\qquad\leq
\calO_\lambda(1)
\ee^{-\beta e_\lambda^D}
\int_0^\beta
\ee^{-t\delta_\lambda^D}\dif{t}
=
\calO_\lambda\br{
\ee^{-\beta e_\lambda^D}
}.
\end{malign}

Similarly, the contribution with the ground-state projection in the
first heat kernel and $Q$ in the second is bounded by
\begin{malign}
&2\left|
\int_0^\beta
\int_{x\in A_{\lambda,+}}
\chi_{\lambda,+}(x)^2
\int_{y\in\partial\Omega_\lambda}
\ee^{-t e_\lambda^D}
\vf_{\lambda,0}^{\Omega_\lambda}(x)
\left[
-\nu_\lambda(y)\cdot
\nabla\vf_{\lambda,0}^{\Omega_\lambda}(y)
\right]
\right.
\\
&\hspace{8em}\left.
\times
\left(
\ee^{-(\beta-t)h_\lambda^{-\Omega_\lambda}}
R Q_\lambda^{\Omega_\lambda}R
\right)(y,-x)
\dif{y}\dif{x}\dif{t}
\right|
\\
&\qquad\leq
\calO_\lambda(1)
\ee^{-\beta e_\lambda^D}
\int_0^\beta
\ee^{-(\beta-t)\delta_\lambda^D}\dif{t}
=
\calO_\lambda\br{
\ee^{-\beta e_\lambda^D}
}.
\end{malign}

Finally, when both heat kernels contain $Q$, Cauchy--Schwarz and the two
$Q$-estimates above give
\begin{malign}
&2\left|
\int_0^\beta
\int_{x\in A_{\lambda,+}}
\chi_{\lambda,+}(x)^2
\int_{y\in\partial\Omega_\lambda}
\left[
-\nu_\lambda(y)\cdot\nabla_y
\left(
\ee^{-t h_\lambda^{\Omega_\lambda}}
Q_\lambda^{\Omega_\lambda}
\right)(x,y)
\right]
\right.
\\
&\hspace{8em}\left.
\times
\left(
\ee^{-(\beta-t)h_\lambda^{-\Omega_\lambda}}
R Q_\lambda^{\Omega_\lambda}R
\right)(y,-x)
\dif{y}\dif{x}\dif{t}
\right|
\\
&\qquad\leq
\calO_\lambda(1)
\int_0^\beta
\ee^{-t\br{e_\lambda^D+\delta_\lambda^D}}
\ee^{-(\beta-t)\br{e_\lambda^D+\delta_\lambda^D}}
\dif{t}
\\
&\qquad=
\calO_\lambda\left(
\beta\ee^{-\beta\br{e_\lambda^D+\delta_\lambda^D}}
\right)
=
\calO_\lambda\br{
\ee^{-\beta e_\lambda^D}
}.
\end{malign}
The last equality uses that
$\beta\ee^{-\beta\delta_\lambda^D}=\calO_\lambda(1)$ as
$\beta\to\infty$.

Thus the sum of all three terms containing at least one
$Q_\lambda^{\Omega_\lambda}$, or its reflected counterpart, is
\eq{
\calO_\lambda\br{
\ee^{-\beta e_\lambda^D}
}.
}
Together with the $PP$ contribution already computed in the main text,
namely
\eq{
2a_\lambda\beta
\ee^{-\beta e_\lambda^D}
q_\lambda,
}
this gives
\eq{
Z_{\lambda,1}(\beta)
=
2a_\lambda
\ee^{-\beta e_\lambda^D}
\br{
\beta q_\lambda
+
\calO_\lambda(1)
},
}
as asserted in \cref{eq:one passage flux asymptotic}.

\section{Verification of the one-well assumptions}
\label{sec:verification-one-well-assumptions}

We record without proof the standard harmonic-approximation facts needed
to justify the one-well assumptions when the Dirichlet boundary depends
on $\lambda$. The relevant condition is that the removed ball remain
large compared with the harmonic scale $\lambda^{-1/2}$.

\begin{prop}[Harmonic approximation for the moving one-well domains]
\label{prop:moving-domain-harmonic-approximation}
Suppose that $V\in C^2(\RR^n;[0,\infty))$ has precisely two
non-degenerate zeros at $\pm d$ and is bounded away from zero outside
fixed neighborhoods of these points. Let $r_\lambda\downarrow0$ satisfy
\eq{
\lambda r_\lambda^2\longrightarrow\infty.
}
For $q\in\Set{d,-d}$, let $D_\lambda(q)$ be the exterior of the ball of
radius $r_\lambda$ about $-q$, with the component containing $q$ taken
when $n=1$, and let $h_\lambda^{D_\lambda(q)}$ be the Dirichlet
restriction of
\eq{
-\Delta+\lambda^2V.
}
Let $\mu_j(q)$ denote the eigenvalues of the harmonic oscillator
\eq{
H_q^{\mathrm{har}}
:=
-\Delta_y
+
\frac12
y\cdot\br{\operatorname{Hess}V(q)}y.
}
Then, for every fixed $j\geq0$,
\eq{
\frac{e_{\lambda,j}^{D_\lambda(q)}}{\lambda}
\longrightarrow
\mu_j(q)
}
uniformly for $q\in\Set{d,-d}$. In particular, for some constants
$c,C>0$ and all sufficiently large $\lambda$,
\eq{
c\lambda
\leq
e_{\lambda,0}^{D_\lambda(q)}
\leq
C\lambda,
\qquad
e_{\lambda,1}^{D_\lambda(q)}
-
e_{\lambda,0}^{D_\lambda(q)}
\geq
c\lambda.
}
\end{prop}

\begin{prop}[Ground-state localization on the intermediate scale]
\label{prop:moving-domain-ground-state-localization}
Under the hypotheses of
\cref{prop:moving-domain-harmonic-approximation}, let
$\vf_{\lambda,0}^{D_\lambda(q)}$ be the positive normalized ground
state. Then, for every fixed $\eta\in(0,\frac12)$,
\eq{
\sup_{x\in B_{\lambda^{-1/2+\eta}}(q)}
\left|
\frac1\lambda
\log
\vf_{\lambda,0}^{D_\lambda(q)}(x)
\right|
\longrightarrow0
}
uniformly for $q\in\Set{d,-d}$. In particular, the choice
$r_\lambda=\lambda^{-1/4}$ satisfies the hypotheses above and verifies
the one-well assumptions used in the main text.
\end{prop}

\newpage
\begingroup
\let\itshape\upshape
\printbibliography

@article{Simon_1984_10.2307/2007072,
	ISSN = {0003486X},
	URL = {http://www.jstor.org/stable/2007072},
	author = {Barry Simon},
	journal = {Annals of Mathematics},
	number = {1},
	pages = {89--118},
	publisher = {Annals of Mathematics},
	title = {Semiclassical Analysis of Low Lying Eigenvalues, II. Tunneling},
	volume = {120},
	year = {1984}
}

@article{Helffer_Sjostrand_1984,
	author = {   B.   Helffer  and    J.   Sjostrand },
	title = {Multiple wells in the semi-classical limit I},
	journal = {Communications in Partial Differential Equations},
	volume = {9},
	number = {4},
	pages = {337-408},
	year  = {1984},
	publisher = {Taylor & Francis},
	doi = {10.1080/03605308408820335},
	
	URL = { 
	https://doi.org/10.1080/03605308408820335
	
	},
	eprint = { 
	https://doi.org/10.1080/03605308408820335
	
	}
	
}

@article{Fefferman2025,
  author   = {Fefferman, Charles L. and Shapiro, Jacob and Weinstein, Michael I.},
  title    = {Lower Bounds on Quantum Tunneling for Excited States},
  journal  = {The Journal of Geometric Analysis},
  year     = {2025},
  volume   = {35},
  number   = {11},
  pages    = {357},
  doi      = {10.1007/s12220-025-02196-w},
  url      = {https://doi.org/10.1007/s12220-025-02196-w},
  issn     = {1559-002X}
}

@incollection{Coleman1985Uses,
  author    = {Coleman, Sidney},
  title     = {The Uses of Instantons},
  booktitle = {Aspects of Symmetry: Selected Erice Lectures},
  publisher = {Cambridge University Press},
  address   = {Cambridge},
  year      = {1985},
  chapter   = {7},
  pages     = {265--350},
  doi       = {10.1017/CBO9780511565045.008}
}

@article{VainshteinEtAl1982ABC,
  author  = {Vainshtein, A. I. and Zakharov, V. I. and
             Novikov, V. A. and Shifman, M. A.},
  title   = {{ABC} of Instantons},
  journal = {Soviet Physics Uspekhi},
  volume  = {25},
  number  = {4},
  pages   = {195--215},
  year    = {1982},
  doi     = {10.1070/PU1982v025n04ABEH004533},
  url     = {https://ufn.ru/en/articles/1982/4/a/}
}

@article{JonaLasinioMartinelliScoppola1981,
  author  = {Jona-Lasinio, Giovanni and Martinelli, Fabio and
             Scoppola, Enzo},
  title   = {New Approach to the Semiclassical Limit of Quantum
             Mechanics. {I}. Multiple Tunnelings in One Dimension},
  journal = {Communications in Mathematical Physics},
  volume  = {80},
  number  = {2},
  pages   = {223--254},
  year    = {1981},
  doi     = {10.1007/BF01213012}
}

@article{CombesDuclosSeiler1983,
  author  = {Combes, Jean-Michel and Duclos, Pierre and Seiler, Ruedi},
  title   = {Convergent Expansions for Tunneling},
  journal = {Communications in Mathematical Physics},
  volume  = {92},
  number  = {2},
  pages   = {229--245},
  year    = {1983},
  doi     = {10.1007/BF01210848}
}

@article{COP93,
  author  = {Cassandro, Marzio and Orlandi, Enza and Presutti, Errico},
  title   = {Interfaces and Typical {Gibbs} Configurations for One-Dimensional {Kac} Potentials},
  journal = {Probability Theory and Related Fields},
  volume  = {96},
  number  = {1},
  pages   = {57--96},
  year    = {1993},
  doi     = {10.1007/BF01195883}
}

@article{BBB08,
  author  = {Bertini, Lorenzo and Brassesco, Stella and Butt{\`a}, Paolo},
  title   = {Dobrushin States in the {$\phi^4_1$} Model},
  journal = {Archive for Rational Mechanics and Analysis},
  volume  = {190},
  number  = {3},
  pages   = {477--516},
  year    = {2008},
  doi     = {10.1007/s00205-008-0151-3}
}

@article{Web10,
  author  = {Weber, Hendrik},
  title   = {Sharp Interface Limit for Invariant Measures of a Stochastic {Allen--Cahn} Equation},
  journal = {Communications on Pure and Applied Mathematics},
  volume  = {63},
  number  = {8},
  pages   = {1071--1109},
  year    = {2010},
  doi     = {10.1002/cpa.20323}
}

@article{OWW14,
  author  = {Otto, Felix and Weber, Hendrik and Westdickenberg, Maria G.},
  title   = {Invariant Measure of the Stochastic {Allen--Cahn} Equation: The Regime of Small Noise and Large System Size},
  journal = {Electronic Journal of Probability},
  volume  = {19},
  number  = {23},
  pages   = {1--76},
  year    = {2014},
  doi     = {10.1214/EJP.v19-2813}
}

@article{RS23,
  author  = {Rezakhanlou, Fraydoun and Seo, Insuk},
  title   = {Scaling Limit of Small Random Perturbation of Dynamical Systems},
  journal = {Annales de l'Institut Henri Poincar{\'e}, Probabilit{\'e}s et Statistiques},
  volume  = {59},
  number  = {2},
  pages   = {867--903},
  year    = {2023},
  doi     = {10.1214/22-AIHP1275}
}

@article{BBDG25,
  author  = {Bertini, Lorenzo and Butt{\`a}, Paolo and Di Ges{\`u}, Giacomo},
  title   = {Asymptotics of the {$\phi^4_1$} Measure in the Sharp Interface Limit},
  journal = {Archive for Rational Mechanics and Analysis},
  volume  = {249},
  number  = {5},
  pages   = {60},
  year    = {2025},
  doi     = {10.1007/s00205-025-02130-y}
}

@article{FeffermanLeeThorpWeinstein2018,
  author  = {Fefferman, Charles L. and Lee-Thorp, James P. and Weinstein, Michael I.},
  title   = {Honeycomb {Schr\"odinger} Operators in the Strong Binding Regime},
  journal = {Communications on Pure and Applied Mathematics},
  volume  = {71},
  number  = {6},
  pages   = {1178--1270},
  year    = {2018},
  doi     = {10.1002/cpa.21735}
}

@article{FeffermanShapiroWeinstein2022,
  author  = {Fefferman, Charles and Shapiro, Jacob and Weinstein, Michael I.},
  title   = {Lower Bound on Quantum Tunneling for Strong Magnetic Fields},
  journal = {SIAM Journal on Mathematical Analysis},
  volume  = {54},
  number  = {1},
  pages   = {1105--1130},
  year    = {2022},
  doi     = {10.1137/21M1429412}
}

@incollection{DimassiSjostrand1999InteractionMatrix,
  author    = {Dimassi, Mouez and Sj{\"o}strand, Johannes},
  title     = {Tunnel Effect and Interaction Matrix},
  booktitle = {Spectral Asymptotics in the Semi-Classical Limit},
  series    = {London Mathematical Society Lecture Note Series},
  volume    = {268},
  pages     = {49--74},
  publisher = {Cambridge University Press},
  address   = {Cambridge},
  year      = {1999},
  doi       = {10.1017/CBO9780511662195.007},
  url       = {https://doi.org/10.1017/CBO9780511662195.007}
}
\endgroup
\end{document}